\documentclass[12pt, draftclsnofoot,onecolumn,letterpaper]{IEEEtran}
\usepackage{inputenc}
\usepackage{cite,acronym,paralist,xspace}
\usepackage[cmex10]{amsmath}
\usepackage{amssymb,amsfonts,amsthm}
\usepackage{algorithmic}
\usepackage{graphicx}
\usepackage{pgf}
\usepackage{textcomp}
\usepackage{xcolor}
\usepackage{subcaption}
\usepackage{stmaryrd} 
\usepackage{dsfont}
\usepackage{balance}
\usepackage{diagbox}
\usepackage{enumitem}
\usepackage{cite}

\usepackage{placeins}
\acrodef{ACDIS}[ACDIS]{Adaptive Communication Decision and Information Systems}
\acrodef{AEP}{Asymptotic Equipartition Property}
\acrodef{AoA}{Angle of Arrival}
\acrodef{AWGN}{Additive White Gaussian Noise}
\acrodef{AVC}[AVC]{Arbitrarily Varying Channel}
\acrodef{BER}{Bit-Error-Rate}
\acrodef{BEC}{Binary Erasure Channel}
\acrodef{BPSK}{Binary Phase-Shift Keying}
\acrodef{BSC}{Binary Symmetric Channel}
\acrodef{BICM}[BICM]{Bit-Interleaved Coded-Modulation}
\acrodef{CDF}[CDF]{Cumulative Distribution Function}
\acrodef{CGF}[CGF]{Cumulant Generating Function}
\acrodef{CLT}[CLT]{Central Limit Theorem}
\acrodef{cq}[c-q]{Classical-Quantum}
\acrodef{CSI}[CSI]{Channel State Information}
\acrodef{DMC}[DMC]{Discrete Memoryless Channel}
\acrodef{DMS}[DMS]{Discrete Memoryless Source}
\acrodef{ERM}[ERM]{Empirical Risk Minimization}
\acrodef{FER}[FER]{Frame Error Rate}
\acrodef{ICA}[ICA]{Independent Component Analysis}
\acrodef{iid}[i.i.d.]{independent and identically distributed}
\acrodef{IoT}[IoT]{Internet of Things}
\acrodef{KKT}[KKT]{Karush-Kuhn Tucker}
\acrodef{LASSO}[LASSO]{Least Absolute Shrinkage and Selection Operator}
\acrodef{LPD}[LPD]{Low Probability of Detection}
\acrodef{LDPC}[LDPC]{Low-Density Parity-Check}
\acrodef{LLMS}[LLMS]{Linear Least Mean Square}
\acrodef{LMS}[LMS]{Least Mean Square}
\acrodef{MAC}[MAC]{multiple-access channel}
\acrodef{MGF}[MGF]{Moment Generating Function}
\acrodef{MLC}[MLC]{Multi-Level Coding}
\acrodef{MLE}[MLE]{Maximum Likelihood Estimate}
\acrodef{MIMO}[MIMO]{Multiple-Input Multiple-Output}
\acrodef{MISO}{Multiple-Input Single-Output}
\acrodef{MSD}[MSD]{Multi-Stage Decoding}
\acrodef{MMSE}[MMSE]{Minimum Mean-Square Error}
\acrodef{PAC}[PAC]{Probably Approximately Correct}
\acrodef{PCA}[PCA]{Principal Component Analysis}
\acrodef{PDF}[PDF]{Probability Density Function}
\acrodef{PMF}[PMF]{Probability Mass Function}
\acrodef{POVM}[POVM]{Positive Operator-Valued Measure}
\acrodef{PVM}[PVM]{Projection-Valued Measure}
\acrodef{PPM}[PPM]{Pulse Position Modulation}
\acrodef{PSD}{Power Spectral Density}
\acrodef{PSK}{Phase Shift Keying}
\acrodef{QKD}{Quantum Key Distribution}
\acrodef{ROC}{Receiver Operating Characteristic}
\acrodef{CVQKD}{Continuous-Variable \ac{QKD}}
\acrodef{QPSK}{Quadrature Phase-Shift Keying}
\acrodef{RF}{Radio-Frequency}
\acrodef{RV}{random variable}
\acrodef{SIMO}{Single-Input Multiple-Output}
\acrodef{SNR}{Signal-to-Noise Ratio}
\acrodef{SVM}[SVM]{Support Vector Machine}
\acrodef{TPCP}{Trace-Preserving Completely-Positive}
\acrodef{wrt}[w.r.t.]{with respect to}
\acrodef{WSS}{Wide Sense Stationary}

\newcommand{\calA}{\mathcal{A}}

\newcommand{\bfC}{\mathbf{C}}
\newcommand{\calC}{\mathcal{C}}

\newcommand{\calD}{\mathcal{D}}
\newcommand{\bbD}{\mathbb{D}}

\newcommand{\bbE}{\mathbb{E}}

\newcommand{\bbH}{\mathbb{H}}

\newcommand{\bbN}{\mathbb{N}}

\newcommand{\calP}{\mathcal{P}}
\newcommand{\bbP}{\mathbb{P}}

\newcommand{\bbR}{\mathbb{R}}

\newcommand{\calS}{\mathcal{S}}

\newcommand{\calT}{\mathcal{T}}

\newcommand{\calX}{\mathcal{X}}

\newcommand{\bfx}{\mathbf{x}}

\newcommand{\calY}{\mathcal{Y}}

\newcommand{\bfy}{\mathbf{y}}

\newcommand{\calZ}{\mathcal{Z}}

\newcommand{\card}[1]{{|{#1}|}}
\newcommand{\abs}[1]{{\left|{#1}\right|}}

\newcommand{\argmin}{\mathop{\text{argmin}}}
\newcommand{\argmax}{\mathop{\text{argmax}}}
\newcommand{\intseq}[2]{[{#1};{#2}]}
\newcommand{\eqdef}{\triangleq}
\newcommand{\E}[2][]{{\mathbb{E}_{#1}\left[#2\right]}}

\renewcommand{\P}[2][]{{\mathbb{P}_{#1}\left(#2\right)}}
\newcommand{\indic}[1]{{\mathbf{1}\{#1\}}}

\newcommand{\D}[3][]{{\mathbb{D}_{#1}}\!\left(#2\,\middle\Vert\,#3\right)} % divergence
\newcommand{\avgI}[1]{{{\mathbb{I}}\!\left(#1\right)}} % mutual information
\newcommand{\avgH}[1]{{\mathbb{H}}\!\left(#1\right)}
\newcommand{\Hb}[1]{{h_b}\left(#1\right)} % Binary entropy function

\renewcommand{\leq}{\leqslant} % redefines <=
\renewcommand{\geq}{\geqslant} % redefines >=

\newcommand{\subparagraph}{}
\usepackage{titlesec}
\usepackage[nameinlink]{cleveref}
\newcounter{relctr} %% <- counter for relations
\everydisplay\expandafter{\the\everydisplay\setcounter{relctr}{0}} %% <- reset every eq
\newcommand\labrel[2]{%
  \begingroup
    \refstepcounter{relctr}%
    \stackrel{\textnormal{(\alph{relctr})}}{\mathstrut{#1}}%
    \originallabel{#2}%
  \endgroup
}
\AtBeginDocument{\let\originallabel\label} %% <- store original definition
\allowdisplaybreaks
\newtheorem{theorem}{Theorem}%
\newtheorem{corollary}[theorem]{Corollary}%
\newtheorem{lemma}[theorem]{Lemma}% 
\newtheorem{definition}[theorem]{Definition}%
\newtheorem{remark}{Remark}
\Crefname{lemma}{Lemma}{Lemmas}
\crefname{lemma}{lemma}{lemmas}

\usepackage{tikz,pgf,pgfplots,xspace,paralist}
\pgfplotsset{compat=1.16}
\usetikzlibrary{arrows,arrows.meta,automata,shapes,calc,intersections}

\title{Rate and Detection-Error Exponent Tradeoff\\ for Joint Communication and Sensing\\ of Fixed Channel States}

\author{Meng-Che Chang, Shi-Yuan Wang, Tuna Erdo\u{g}an, and Matthieu R. Bloch\thanks{The authors are with the School of Electrical and Computer Engineering, Georgia Institute of Technology, Atlanta, GA (email:\{mchang301, shi-yuan.wang, etuna3\}@gatech.edu, matthieu.bloch@ece.gatech.edu). This work has been supported by the National Science Foundation (NSF) under grants 1955401 and 2148400 as part of the Resilient \& Intelligent NextG Systems (RINGS) program. Parts of the results have been presented at the 2022 IEEE International Symposium on Joint Communications \& Sensing~\cite{Chang2022Rate}.}}

\begin{document}
\maketitle

\begin{abstract}
  We study the information-theoretic limits of joint communication and sensing when the sensing task is modeled as the estimation of a discrete channel state fixed during the transmission of an entire codeword. This setting captures scenarios in which the time scale over which sensing happens is significantly slower than the time scale over which symbol transmission occurs. The tradeoff between communication and sensing then takes the form of a tradeoff region between the rate of reliable communication and the state detection-error exponent. We investigate such tradeoffs for both mono-static and bi-static scenarios, in which the sensing task is performed at the transmitter or receiver, respectively. In the mono-static case, we develop an exact characterization of the tradeoff in open-loop, when the sensing is not used to assist the communication. We also show the strict improvement brought by a closed-loop operation, in which the sensing informs the communication. In the bi-static case, we develop an achievable tradeoff region that highlights the fundamentally different nature of the bi-static scenario. Specifically, the rate of communication plays a key role in the characterization of the tradeoff and we show how joint strategies, which simultaneously estimate message and state, outperform successive strategies, which only estimate the state after decoding the transmitted message.
\end{abstract}

\section{Introduction}
\label{sec:intro}

A core feature envisioned for the next generation of mobile networks is the convergence of communication and sensing~\cite{Paul2017Survey,Fettweis2021Joint} (also known as integrated communication and sensing), motivated in part by the need to offer detection and localization capabilities to interconnected devices interacting with the real world (robots, UAVs, etc.). This convergence is also enabled by the shift of communication frequencies towards the mmWave part of the spectrum, which allows a single \ac{RF} hardware to perform both communication and sensing. While joint communication and sensing has already attracted interest in the context of joint communication and radar~\cite{Chiriyath2016Inner,Zheng_2019,Liu2020Joint}, theoretical, algorithmic, and hardware-related challenges remain and must be addressed to assess the true potential of joint communication and sensing~\cite{Fettweis2021Joint,Wild_2021}.

The objective of the present work is to further the theoretical understanding of the information-theoretic limits of joint communication and sensing, and specifically to better understand the tradeoffs incurred by a joint operation. While information-theoretic models often abstract fine channel modeling aspects, they provide valuable insights to identify the regimes in which tradeoffs exist and to quantify their severity. In particular, information-theoretic limits of joint communication and sensing are naturally approached from the perspective of joint channel transmission and channel state estimation, where the state simultaneously represents the quantity of interest for sensing and affects the communication. Early works~\cite{Kim2008,Choudhuri2011Causal,Zhang_2011} have considered state-dependent channel models with \ac{iid} channel states, in which the encoder attempts to simultaneously communicate and facilitate the estimation of the state at the receiver, and have leveraged rate-distortion theory to characterize the optimal tradeoff between communication rate and state reconstruction accuracy. Recently,  \cite{Kobayashi2018Joint,Ahmadipour2022Information} have revisited the model of~\cite{Zhang_2011} by shifting the task of estimating the state from the receiver to the transmitter using generalized feedback, a situation more in line with the scenarios envisioned in the context of joint communication and sensing, and characterized again the optimal tradeoff between rate and average state distortion. Extensions to multiple-access channels~\cite{Kobayashi_2019} and broadcast channels~\cite{Ahmadipour2021Joint,Ahmadipour2022Information} have been investigated, as well, although exact characterizations of the tradeoff remain elusive. Recent works have also considered secure joint communication and sensing to quantify and investigate the intrinsic information leakage associated to a joint operation~\cite{Su2021Secure,Ren2022Robust,Guenlue2022Secure}.

A common feature of~\cite{Zhang_2011,Kobayashi2018Joint,Kobayashi_2019,Ahmadipour2021Joint,Ahmadipour2022Information,Guenlue2022Secure} is that the \ac{iid} nature of the channel state precludes any prediction. Consequently, state detection and estimation strategies are open-loop and the tradeoff between communication and sensing reduces to a resource allocation problem, in which the choice of a channel input distribution dictates the tradeoff.\footnote{The situation is more nuanced when introducing secrecy constraints as in~\cite{Guenlue2022Secure} as feedback is known to improve secrecy.} Furthermore, the \ac{iid} model is not well-suited to applications in which the physical phenomena sensed, such as the presence or absence of an obstacle that would disrupt line-of-sight communication, evolve on a time scale that is much slower than the time scale at which communication symbols are transmitted. To address such applications, one may instead adopt a model in which the channel state remains \emph{constant} over the block-length used for communication. Except in rare cases, estimating the channel state is then always possible so that the tradeoff between communication and sensing appears between the rate of communication and the accuracy of the state estimation. For the case of uncountably infinite channel states parameterized by continuous variables, accuracy can be captured by the Cramer-Rao bound~\cite{Xiong2022Flowing}. For the case of finitely many channel states, accuracy may be captured by the asymptotic state detection error exponent, as studied thereafter and already reported in our preliminary results~\cite{Chang2022Rate} concurrently with~\cite{Joudeh2022Joint,Wu2022Joint}.

Two kinds of joint communication and sensing models are considered in this paper, namely, mono-static and bi-static models, where we borrow these terminologies from mono-static and bi-static radars. In both models, there is a transmitter-receiver pair attempting to convey messages without exact knowledge of the channel state. The difference between these two models is that sensing is performed at the transmitter in the mono-static model and at the receiver in the bi-static model. Since the transmitter always knows the codeword, sensing in the mono-static model can be done coherently with the knowledge of the transmitted waveform.  In contrast, the receiver in the bi-static model needs to simultaneously sense the channel and decode the message (or waveform).

In what follows, we study the rate/detection-error exponent tradeoff for both mono-static and bi-static models~\cite{Paul2017Survey} with fixed channel states. Accordingly, our approach draws on the extensive literature on controlled sensing~\cite{Nitinawarat2013,Naghshvar2013a,Naghshvar2013} and channel estimation with pilot sequences~\cite{Tchamk2006use,Mahajan_2012}. Our specific contributions are:
%~\cite{Zhang_2011,Ahmadipour2021Joint,Ahmadipour2021,Chang2022Rate}. 
\begin{inparaenum}
\item we characterize the exact rate-detection exponent tradeoff for mono-static open-loop joint communication and sensing;\footnote{The concurrent works~\cite{Joudeh2022Joint,Wu2022Joint} developed characterizations similar to ours in more restricted settings; the results therein are subsumed by Theorem~\ref{thm:mono-capacity} of the present document and~\cite[Theorem 1]{Chang2022Rate}.}
\item we show a strict improvement of the rate-detection exponent tradeoff for mono-static closed-loop joint communication and sensing through learning the channel state and adapting the channel code;
\item we provide a partial characterization of the rate-detection exponent tradeoff for bi-static open-loop joint communication and sensing, and show that jointly detecting the state and decoding the code strictly outperforms the naive successive method, in which the detection of the state follows the decoding of the channel code.
\end{inparaenum}
We further illustrate the results with simple examples that capture the essence of realistic joint communication and sensing scenarios.

The remaining of the paper is organized as follows. We introduce the mono-static and bi-static models for joint communication and sensing in Section~\ref{sec:joint-comm-sens-models}. We present our main results in Section~\ref{sec:main-results}, along with numerical examples illustrating the joint communication and sensing tradeoffs. %The specific results are the following:
%\cmnt{
%\begin{itemize}
%    \item an exact closure region for the mono-static joint communication and sensing rate/detection-error exponent pairs with open-loop strategies (Theorem~\ref{thm:mono-capacity});
%    \item an inner region for the mono-static joint communication and sensing rate/detection-error exponent pairs with closed-loop strategies (Theorem~\ref{thm:2});
%    \item inner regions for the bi-static joint communication and sensing rate/detection-error exponent pairs with successive (Theorem~\ref{thm:bi-static-result}) and joint strategies (Theorem~\ref{thm:bi-static-result2}), respectively.
%\end{itemize}
%}
We relegate all proofs to Section~\ref{sec:proofs} to streamline the presentation.

\section{Notation}
\label{sec:notation}
For any set $\Omega$, the indicator function is defined as $\mathbf{1}(\omega\in\Omega) = 1$ if $\omega\in \Omega$ and $0$ otherwise.
For any discrete set $\calX$, $\mathcal{P}_{\calX}$ is the set of all probability distributions on $\calX$. For $n\in\bbN^*$, a sequence of length $n$ is implicitly denoted $\mathbf{x}\eqdef(x_1,\cdots,x_n)\in \calX^n$, while $x^i\eqdef (x_1,\cdots,x_i)\in \calX^i$ denotes a sequence of length $i$, and $x_j^i=(x_j,\cdots,x_i)$ is a sub-sequence of $\mathbf{x}$.  For $\bfx\in\calX^n$, $\hat{p}_{\mathbf{x}}$ denotes the type of $\mathbf{x}$, i.e., $\hat{p}_{\mathbf{x}}(x) = \frac{1}{n}\sum_{i=1}^n\indic{x_i=x}$. For any type $P$, $\mathcal{T}_{P}^n$ is the corresponding type class, i.e., the set of all sequences $\mathbf{x}\in\calX^n$ such that $\hat{p}_{\mathbf{x}}=P$. $\mathcal{P}_{\calX}^n$ is the set of all possible types for length $n$ sequences in $\calX^n$.
Let $\calP_{\calY|\calX}$ be the set of all conditional probabilities of $Y\in\calY$ given $X\in\calX$. Given a sequence $\bfx\in \mathcal{X}^n$ and $\bfy\in\mathcal{Y}^n$, we define $\hat{p}_{\bfy|\bfx}$ as the empirical conditional type, i.e., $\hat{p}_{\bfy|\bfx}(b|a) = \sum_{i=1}^n \indic{x_i=a,y_i=b}/\sum_{i=1}^n \indic{x_i=a}$ for all $a\in\mathcal{X}$ such that $\hat{p}_{\bfx}(a)>0$ and $b\in\mathcal{Y}$. %$N_{y^n,x^n}(b,a)\triangleq \sum_{i=1}^n \indic{y_i=b,x_i=a}$ for all $a\in\mathcal{X}$ and $b\in\mathcal{Y}$. 
Let $\calP_{\calY|\calX}^n$ be the set of all conditional types for length $n$ sequences $\bfx\in \calX^n$ and $\bfy\in \calY^n$.
For any conditional type $P_{Y|X}\in\mathcal{P}_{\mathcal{Y}|\mathcal{X}}^n$, we also define $\mathcal{T}_{\cdot|\bfx}(P_{Y|X})$ as the conditional type class of $P_{Y|X}$, i.e., the set of sequences $\bfy\in\mathcal{Y}^n$ such that $\hat{p}_{\bfy|\bfx}=P_{Y|X}$, and define $\calT_{\bfy|\cdot}(P_{Y|X})$ as the set of sequences $\bfx\in\calX^n$ such that $\hat{p}_{\bfy|\bfx}=P_{Y|X}$.
Given two conditional distributions $W_{Y|X}$ and $P_{Y|X}$, we set
$%$\begin{align*}
\left|W_{Y|X}-P_{Y|X}\right|_{\infty}\triangleq \max_{a\in\mathcal{X},b\in\mathcal{Y}} |W_{Y|X}(b|a)-P_{Y|X}(b|a)|.
$ %\end{align*}
We let $\avgH{P_X}\eqdef-\sum_{x\in\calX}P_X(x)\log P_X(x)$ be the entropy of $X\sim P_X$. If $W_{Y|X}$ is a conditional distribution on $Y\in\calY$ given $X\in\calX$, $\avgH{W_{Y|X}\middle|P_X}\eqdef\E[P_X]{\avgH{W_{Y|X}(\cdot|X)}}$ is the conditional entropy of $W_{Y|X}$ given an input distribution $P_X$ and $\mathbb{I}(P_X,W_{Y|X})\eqdef\avgH{W_{Y|X}\circ P_X}-\avgH{W_{Y|X}\middle|P_X}$ is the mutual information between $X$ and $Y$, where $X\sim P_X$ and $Y\sim P_X\circ W_{Y|X}\eqdef\sum_{x}P_X(x)W_{Y|X}(\cdot|x)$. The relative entropy between $P_{Y|X}$ and $W_{{Y|X}}$ given an input distribution $P_X$ is
$\D{P_{Y|X}}{W_{Y|X}\middle|P_X}\eqdef\E[P_X]{\D{P_{Y|X}(\cdot|X)}{W_{Y|X}(\cdot|X)}}$.
Throughout the paper, $\log$ is \ac{wrt} base $e$, and therefore all the information quantities should be understood in \emph{nats}. Moreover, for $a,b\in\bbR$ such that $\lfloor a\rfloor\leq \lceil b\rceil$, we define $\intseq{a}{b}\eqdef\{\lfloor a\rfloor, \lfloor a\rfloor+1, \cdots, \lceil b\rceil-1, \lceil b\rceil\}$; otherwise $\intseq{a}{b}\eqdef\emptyset$. In addition, for any $x\in\bbR$, we let $\abs{x}^+$ denote $\max(x,0)$.

\section{Joint communication and sensing models}
\label{sec:joint-comm-sens-models}

\subsection{Mono-Static Model}
\label{sec:mono-static}

\begin{figure}[htb]
    \centering
    \scalebox{0.8}{\includegraphics[]{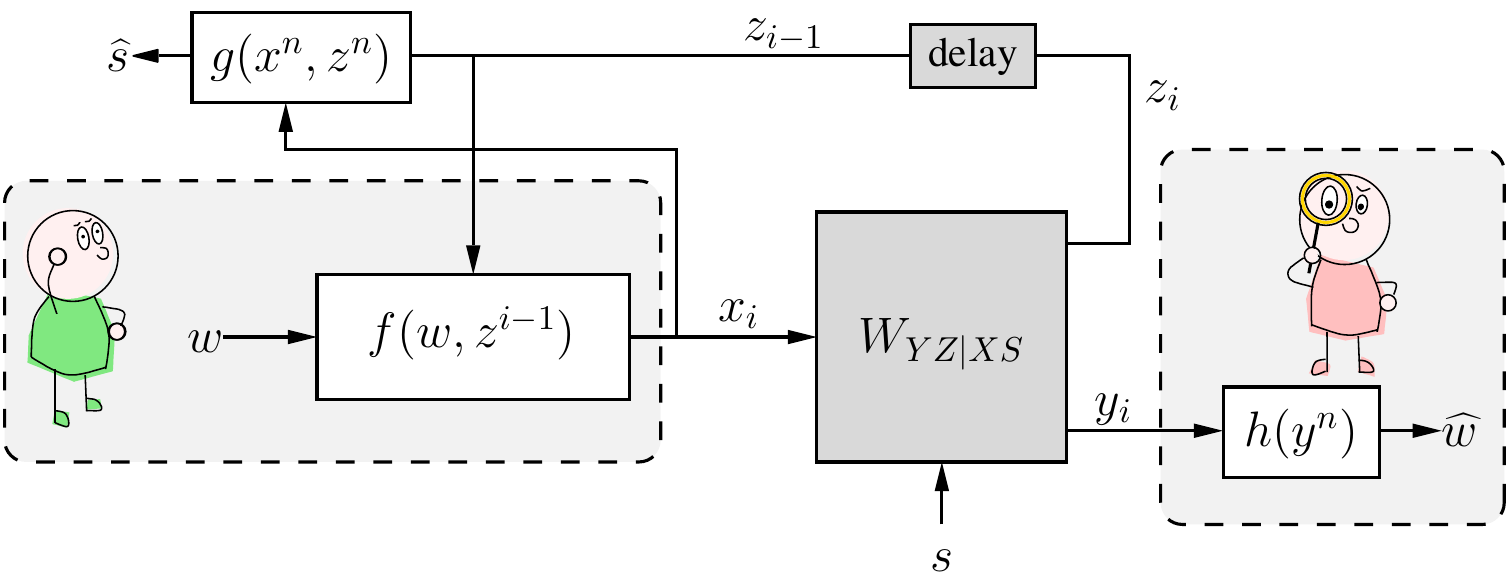}}
    \caption{Mono-static joint communication and sensing model.}
    \label{fig:model}
\end{figure}
The mono-static joint communication and sensing model is illustrated in Fig.~\ref{fig:model}, in which a transmitter attempts to communicate with a receiver over a state-dependent \ac{DMC}, also known as a compound channel, while simultaneously probing the channel state in a strictly causal manner through a sensing channel. Specifically, the transmitter encodes a uniformly distributed message $W\in\intseq{1}{M}$ into a length $n$ codeword $X^n$, of which symbols are transmitted over a \ac{DMC} with transition probability $W_{YZ|XS}$. The state $S$, a priori unknown to both the transmitter and the receiver, is assumed to be \emph{fixed} during the whole duration of the transmission and takes value in a finite set $\calS$.
The transmitter has the ability to estimate the channel state by using past observations obtained from the output $Z$ of the \ac{DMC}, allowing it to adapt its transmission in an online fashion. We assume in this paper that Chernoff information between channels $W_{Z|X,s}$ and $W_{Z|X,s'}$ is non-zero for all $s\neq s'$, i.e.,
$$ \max_{P_X\in\mathcal{P}_{\mathcal{X}}}\max_{\ell\in [0,1]} -\sum_{x}P_X(x)\log \left(\sum_{z} W_{Z|X,s}(z|x)^{\ell}W_{Z|X,s'}(z|x)^{1-\ell}\right) > 0.$$
Formally, the encoder consists of a set of functions 
\begin{align*}
    f_i^{(m)}:\intseq{1}{M}\times\calZ^{i-1}\to\calX:(w,z^{i-1})\mapsto x_i \eqdef f_i^{(m)}(w,z^{i-1})
\end{align*}
defined for every $i\in\intseq{1}{n}$, while the state estimator is a function
\begin{align*}
    g^{(m)}:\calX^{n}\times \calZ^{n}\to\calS:(x^{n}, z^{n})\mapsto \hat{s}.
\end{align*}
The message decoder is a function
\begin{align*}
    h^{(m)}:\calY^n\to\intseq{1}{M}:y^n\mapsto\hat{w}.
\end{align*}
A code $\calC^{(m)}$ then consists of the tuple $(\{f_i^{(m)}\}_{i\in [1;n]},g^{(m)},h^{(m)})$, as well as the implicitly defined associated message set $\intseq{1}{M}$. 

\subsection{Bi-Static Model}
\label{sec:bi-static-model}

\begin{figure}[htb]
    \centering
    \scalebox{0.8}{\includegraphics[]{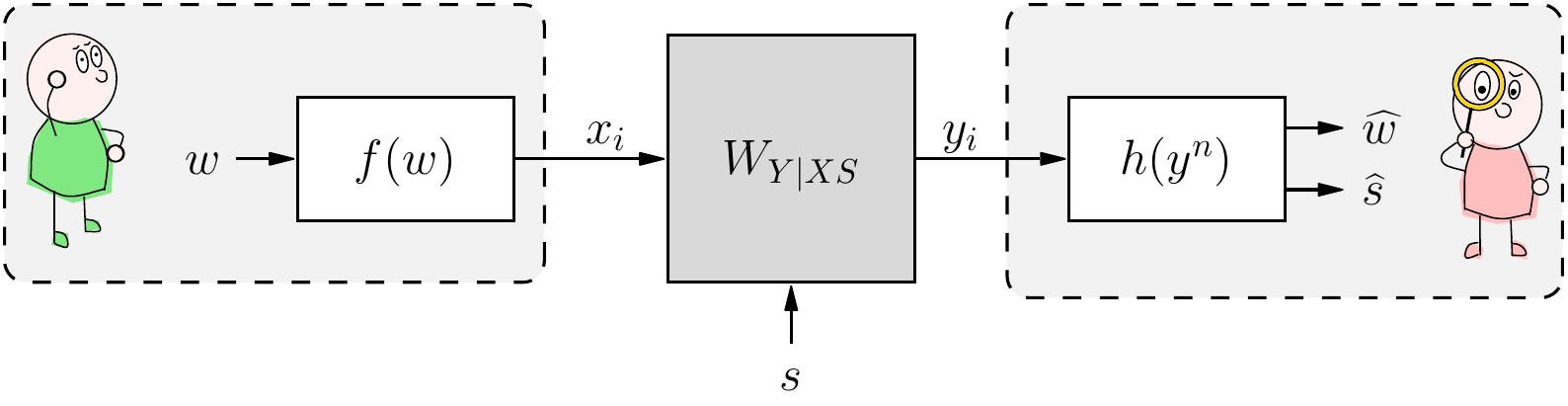}}
    \caption{Bi-static joint communication and sensing model.}
    \label{fig:model_bi-static}
  \end{figure}
  
As illustrated in Fig.~\ref{fig:model_bi-static}, the bi-static joint communication and sensing model differs from the mono-static one in that the receiver should simultaneously sense the state and decode the message. Specifically, the transmitter encodes a uniformly distributed message $W\in[1;M]$ into a length $n$ codeword $X^n$, of which symbols are transmitted over a DMC with transition probability $W_{Y|XS}$. We again assume that the a priori unknown state $S$ is \emph{fixed} during the whole duration of the transmission and takes value in a finite set
  $\mathcal{S}$. The encoder is then defined as 
$$
    f^{(b)}:[1;M] \to \mathcal{X}^n:w\mapsto \mathbf{x}_w,
$$
while the decoder and the state estimator are the functions
$$
    g^{(b)}:\mathcal{Y}^n \to \mathcal{S}:y^n\mapsto \hat{s},
$$
and 
$$
    h^{(b)}:\mathcal{Y}^n \to \intseq{1}{M}:y^n\mapsto \hat{w}.
$$
The code $\mathcal{C}^{(b)}$ in the bi-static model then consists of the tuple $(f^{(b)},g^{(b)},h^{(b)})$, as well as the message set $[1;M]$. 

\begin{remark}
  A key difference between the mono-static model and the bi-static model is that the former reduces the coupling of performance between communication and sensing to the type of the transmitted codewords. Specifically, the mono-static model endows the state estimator with full knowledge of the transmitted codeword, so that correct decoding is irrelevant to the state detection performance. In contrast, the bi-static model requires the receiver to perform a joint estimation of the codeword and of the state. 
\end{remark}

\begin{remark}
The model of Fig.~\ref{fig:model} differs from the ones in~\cite{Zhang_2011,Ahmadipour2022Information}, in which the state is \ac{iid} and changing from symbol to symbol. Our model captures a scenario in which the coherence time of the state is much longer than the duration of a transmission. Since the state does not change during the transmission, the transmitter can gradually obtain an accurate estimation about the state and is able to adapt the transmission scheme according to the estimated channel state. As a result our model also captures the ability to adapt to the channel state in an online fashion, while the models in~\cite{Zhang_2011,Ahmadipour2022Information} only allow for an offline adaptation based on a target rate/distortion pair. Neither model supersedes the other and both capture scenarios that could be relevant to next generation communication networks.
\end{remark}

\subsection{Performance Evaluation Metrics}
The performance of the system is measured in terms of the asymptotic rate of reliable communication and asymptotic detection-error exponent. Formally, we define the communication-error probability and the detection-error probability for both mono-static and bi-static cases as follows
%\cmnt{[For now, I only have result for average error probability in bi-static model. Should we just use averaged error probability for consistency.]}
\begin{itemize}
    \item Mono-static case:
        \begin{align}
    P_{\text{c}}^{(n)} &\triangleq \max_{s\in\calS}\max_{w\in[1;M]}\label{eq:mono-pc} \mathbb{P}(h^{(m)}(Y^n)\neq w|W=w,S=s),\\
     P_{\text{d}}^{(n)} &\eqdef \max_{s\in\calS}\max_{w\in\intseq{1}{M}} \mathbb{P}(g^{(m)}(X^n,Z^n)\neq s|W=w,S=s).\label{eq:mono-pd}%\max_{s\in\calS}\mathbb{P}(g(Z^n)\neq s|S=s)
     \end{align}
     \item Bi-static case:
    \begin{align}
        P_c^{(n)} &\triangleq \max_{s\in\calS}\max_{w\in[1;M]}\label{eq:bi-pc} \mathbb{P}(h^{(b)}(Y^n)\neq w|W=w,S=s),\\
        P_{\text{d}}^{(n)} &\eqdef \max_{s\in\calS}\max_{w\in[1;M]} \mathbb{P}(g^{(b)}(Y^n)\neq s|W=w,S=s).\label{eq:bi-pd}%\max_{s\in\calS}\mathbb{P}(g(Z^n)\neq s|S=s)
    \end{align}

\end{itemize}
The rate and the detection-error exponent for both mono-static and bi-static models are
\begin{align*}
    R^{(n)}&\eqdef\frac{1}{n}\log M\text{ and }
    E_{\text{d}}^{(n)} \triangleq -\frac{1}{n} \log P_{\text{d}}^{(n)},
\end{align*} 
respectively.
%  is composed of the message set, the encoder of the message, the estimator function and the message decoder. With a slight abuse of notation, we denote $|C|=M$ as the number of codewords available. 
\begin{definition}[Achievability]
A rate/detection-error exponent $(R,E)$ is achievable in the mono-static (bi-static) joint communication model if for any $s\in\calS$ and any $\epsilon>0$, there exist a sufficiently large $n$ and a code $\calC^{(m)}$ ($\calC^{(b)}$) of length $n$ such that
\begin{align}
    P_{\textnormal{c}}^{(n)} &\leq \epsilon, \\
    E_{\textnormal{d}}^{(n)} &\geq E-\epsilon,\\
    R^{(n)} &\geq R-\epsilon.
\end{align}
\end{definition}
When the encoder does not perform any online adaptation so that $f_i^{(m)}:[1;M]\mapsto \mathcal{X}$ is independent of the observation $Z^{i-1}$, the scheme is called \emph{open-loop}. On the other hand, if the encoder utilizes feedback information, the scheme is called \emph{closed-loop}. 
Our objective is to characterize the set of all achievable rate/detection-error exponent pairs by open-loop strategies in both mono-static and bi-static models as well as the set of all achievable rate/detection-error exponent pairs by closed-loop strategies in the mono-static model.
\begin{definition} We define $\mathbf{C}_{\textnormal{open}}^{(m)}$ and $\mathbf{C}^{(b)}$ as the closure of all achievable rate/detection-error exponent pairs by open strategies in the mono-static and bi-static model, respectively. Similarly, we define $\mathbf{C}_{\textnormal{close}}^{(m)}$ as the closure of all achievable rate/detection-error exponent pairs by closed-loop strategies in the mono-static model.

\end{definition}

%\cmnt{Note that in the definition of $s$-Achievability, the code only need to be designed for the specific channel with some state $s$. This implies that the codes which achieve $(R_1,E_1)$ and $(R_2,E_2)$ in the state $1$ and state $2$ channels, respectively, are different. 
%The notion of $s$-achievability is useful when someone only care about the performance of the channel with the specific state. }

%\begin{remark}
%There is an asymmetry in our definitions of the probability of errors in~\eqref{eq:pc} and~\eqref{eq:pd}. While~\eqref{eq:pc} includes a maximum over a possible messages $w\in\intseq{1}{M}$,~\eqref{eq:pd} includes an average over all possible messages. This subtlety is only used in our converse proof for Theorem~\ref{thm:capacity} in Section~\ref{sec:converse} and is required to avoid having the detection performance dictated by a codeword whose type is not representative of the code. 
%\end{remark}

\section{Main results}
\label{sec:main-results}
\subsection{Mono-Static Model}
\label{sec:mono-static-result}

We first restrict ourselves to open-loop schemes, which provide a baseline for assessing the usefulness of adaptation. For simplicity, we denote in this case the encoder that maps a message $w$ to a codeword of $n$ symbols by ${f}^{(m)}:[1;M]\mapsto \calX^n$. The following theorem provides an exact characterization of $\mathbf{C}_{\text{open}}^{(m)}$.

\begin{theorem}
  \label{thm:mono-capacity}
  The closure of all achievable joint communication and sensing rate/detection error exponent pairs for mono-static open-loop schemes is
\begin{align}
\mathbf{C}_{\textnormal{open}}^{(m)} &= \bigcup_{P_X\in\mathcal{P}_{\calX}}\left\{\begin{array}{l}(R,E)\in\bbR_+^2:\\
R\leq \min_{s\in\calS} \mathbb{I}(P_X,W_{Y|X,s}) \label{eqn:thm1_1}\\
E\leq \phi(P_X)
\end{array}\right\}
\end{align}
where 
\begin{small}
\begin{align}
    \phi(P_X) &= \min_{s\in\calS}\min_{s'\neq s}\max_{\ell\in [0,1]} -\sum_{x}P_X(x)\log \left(\sum_{z} W_{Z|X,s}(z|x)^{\ell}W_{Z|X,s'}(z|x)^{1-\ell}\right).\label{eq:phi}
\end{align}
\end{small}
\end{theorem}
\begin{proof}
  See Section~\ref{sec:achievability} and Section~\ref{sec:converse}.
\end{proof}
A couple of comments are in order. First, since open-loop schemes do not exploit the information about the state contained in past noisy observations of the state, achievable rates are necessarily upper bounded by the compound channel capacity~$\max_{P_X}\min_{s\in\calS} \mathbb{I}(P_X,W_{Y|X,s})$. This is a weakness of all open-loop schemes. Second, because of the open-loop nature of the coding schemes, the interplay between communication and sensing is captured by the choice of a distribution $P_X$ that governs the empirical statistics of the codewords and is set \emph{offline}. This is similar to what is obtained in other information-theoretic approaches based on rate-distortion~\cite{Zhang_2011,Ahmadipour2022Information}. 

The results in \cite{Joudeh2022Joint} are also special cases of  Theorem~\ref{thm:mono-capacity}. In \cite{Joudeh2022Joint}, the authors consider a mono-static joint communication and sensing model in which $W_{YZ|X,S}=W_{Y|X}W_{Z|X,S}$, i.e., the communication channel is irrelevant to the state. One of the channel models in \cite{Joudeh2022Joint} is a binary setting in which $\calX=\calZ=\calS=\{0,1\}$ and  $W_{Z|X,S}=W_{Z|X\cdot S}$, i.e., at each time $t$ the state-estimator obtains $Z_t = X_t \cdot S \oplus N_t$, where $\oplus$ denotes the modulo-2 sum and $N_t\sim \text{Ber}(q)$ for some $0<q<1$. By specializing Theorem~\ref{thm:mono-capacity}, one recovers \cite[Theorem 1]{Joudeh2022Joint} as follows.
    
\begin{corollary}
Let $\calX=\calY=\calZ=\calS=\{0,1\}$. 
When the mono-static joint communication model satisfies $W_{YZ|XS}=W_{Y|X}W_{Z|X\cdot S}$, where $W_{Y|X}$ and $W_{Z|X\cdot S}$ are  binary symmetric channels with cross over probability $p$ and $q$, respectively, then 
\begin{align}
    \mathbf{C}_{\textnormal{open}}^{(m)} = \bigcup_{\alpha\in[0.5,1]} \left\{\begin{array}{l}(R,E)\in\bbR_+^2:\\
R\leq \bbH(\textnormal{Ber}(\alpha * p)) - \bbH(\textnormal{Ber}(p))\\
E\leq \alpha \bbD\left(\textnormal{Ber}(0.5)||\textnormal{Ber}(q)\right),
\end{array}\right\}
\end{align}
where $\alpha*p = \alpha(1-p)+(1-\alpha)*p$. 
\label{cor:3}
\end{corollary}
\begin{proof}
Since, $W_{YZ|XS}=W_{Y|X}W_{Z|X\cdot S}$, the achievable rate is irrelevant to the state, and for each $P_X\sim\textnormal{Ber}(\alpha)$ we have
\begin{align*}
    \min_{s\in\calS} \mathbb{I}(P_X,W_{Y|X,s}) &= \mathbb{I}(P_X,W_{Y|X}) \\
    &= \bbH(P_X\circ W_{Y|X}) - \bbH(W_{Y|X}|P_X) \\
    &=  \bbH(\text{Ber}(\alpha\ast p)) - \bbH(\text{Ber}(p)).
\end{align*}
Moreover, for such $P_X$, $\phi(P_X)$ can be calculated as follows. 
\begin{align}
    \phi(P_X) &= \min_{s\in\calS}\min_{s'\neq s}\max_{\ell\in [0,1]} -\sum_{x}P_X(x)\times\log \left(\sum_{z} W_{Z|X,s}(z|x)^{\ell}W_{Z|X,s'}(z|x)^{1-\ell}\right)
    \nonumber\\
    &= \max_{\ell\in[0,1]}-P_X(1)\log \left(q^{\ell}(1-q)^{1-\ell}+(1-q)^{\ell}q^{1-\ell}\right)\label{enq:cor4_2}\\
    &= -P_X(1)\log \left(q^{1/2}(1-q)^{1/2}+(1-q)^{1/2}q^{1/2}\right)\label{enq:cor4_3}\\
    &= \alpha \bbD\left(\text{Ber}(0.5)||\text{Ber}(q)\right) \nonumber
    ,
\end{align}
where (\ref{enq:cor4_3}) follows from the fact that $\ell=1/2$ maximizes \eqref{enq:cor4_2}. The corollary follows by taking the union over all $P_X$. 
\end{proof}

\begin{remark}
  Note that in the setting considered in Corollary~\ref{cor:3}, transmitting $X=0$ does not help the performance of the state-estimation, and the detection-error exponent is a monotonously increasing function of the weight of the codeword. 
\end{remark}

We also observe that in some cases, there is no tradeoff between maximizing the communication capacity and the detection-error exponent.
\begin{corollary}
If there exists $x_0\in\calX$ such that for all $x\in\calX$ there exists a permutation $\pi_x$ on $\calZ$ such that for every $s\in\calS$
\begin{align}
    W_{Z|X, s}(z|x)=W_{Z|X, s}(\pi_x(z)|x_0),
\end{align}
then 
\begin{align}
\mathbf{C}_{\textnormal{open}}^{(m)} &= \left\{\begin{array}{l}(R,E)\in\bbR_+^2:\\
R\leq \max_{P_X}\min_{s\in\calS} \mathbb{I}(P_X,W_{Y|X,s})\\
E\leq \max_{P_X}\phi(P_X)
\end{array}\right\}
\end{align}
where $\phi(\cdot)$ is defined in~\eqref{eq:phi}.
 %is shown, as the type of input does not affect the error exponent.
 \label{cor:1}
\end{corollary}

\begin{proof}
  For every $x\in\calX\backslash\{x_0\}$,
  \begin{align}
    \sum_{z\in\calZ} W_{Z|X, s}(z|x)^\ell W_{Z|X, s^\prime}(z|x)^{1-\ell}&=\sum_{z\in\calZ} W_{Z|X, s}(\pi_x(z)|x_0)^\ell W_{Z|X, s^\prime}(\pi_x(z)|x_0)^{1-\ell}\\
    &=\sum_{\pi^{-1}_{x}(z\prime)\in\calZ}W_{Z|X,s}(z^\prime|x_0)^\ell W_{Z|X,s^\prime}(z^\prime|x_0)^{1-\ell}\\
    &=\sum_{z\in\calZ}W_{Z|X,s}(z|x_0)^\ell W_{Z|X,s^\prime}(z|x_0)^{1-\ell}.
  \end{align}
Thus, we know that the detection-error exponent is invariant to the input type under this scenario.
\end{proof}
In other words, when the channel satisfies certain symmetry conditions, there is no tradeoff between rate and detection-error exponent and one simultaneously achieves the optimal communication rate and the optimal detection performance.
One of the compound channel families that falls into such a category is the set of \acp{BSC}. The maximal detection-error exponent and the compound capacity are then simultaneously  achieved with a uniform input distribution.\smallskip

We now turn our attention back to the characterization of \emph{closed-loop schemes}, which exploit the feedback to adapt to the state. The next theorem characterizes an inner bound of the set $\mathbf{C}^{(m)}_{\text{closed}}$. 

\begin{theorem}
  The closure of all achievable joint communication and sensing rate/detection error exponent pairs for mono-static closed-loop schemes satisfies
  \begin{align}
    \mathbf{C}^{(m)}_{\textnormal{closed}}  \supseteq\bigcup_{\{P_{X,s''}\}_{s''\in\mathcal{S}}\in\left(\mathcal{P}_{\calX}\right)^{|\mathcal{S}|}} 
  \left\{\begin{array}{l}
           (R,E)\in \bbR_+^2: \\
           R \leq  \min_{s\in\mathcal{S}}\mathbb{I}(P_{X,s},W_{Y|X,s})\\
           E \leq \min_{s\in\mathcal{S}}\phi(P_{X,s})
         \end{array}
  \right\}
\end{align}
where 
the notation $\bigcup_{\{P_{X,s''}\}_{s''\in\mathcal{S}}\in\left(\mathcal{P}_{\calX}\right)^{|\mathcal{S}|}}$ means that we are taking the union over all possible $|\mathcal{S}|$-tuples of probability distributions in $\mathcal{P}_{\mathcal{X}}$ and $\left(\mathcal{P}_{\mathcal{X}}\right)^{ |\mathcal{S}|}$ is the set of tuples of $|\mathcal{S}|$ elements in $\mathcal{P}_{\mathcal{X}}$.  
\label{thm:2}
\end{theorem}
Theorem~\ref{thm:2} is obtained by considering a simple strategy in which the transmitter learns the state, informs the receiver, and uses a code adapted to the learned channel state. The exact characterization of the optimal tradeoffs for closed-loop schemes remains elusive and presents non-trivial challenges, chief among them the absence of a known optimal detection error-exponent for multi-hypothesis controlled sensing~\cite{Nitinawarat2013}. One can conclude that the maximal achievable detection-error exponent characterized by Theorem~\ref{thm:2} is identical to that of the open-loop strategy by observing the following equality 
\begin{align*}
    \max_{\{P_{X,s''}\}_{s''\in\mathcal{S}}\in\left(\mathcal{P}_{\calX}\right)^{|\mathcal{S}|}}  \min_{s\in\mathcal{S}} \phi(P_{X,s}) = \max_{P_X\in\mathcal{P}_{\mathcal{X}}} \phi(P_X).
\end{align*}

Therefore, the region characterized in Theorem~\ref{thm:2} is sub-optimal because it is already shown in \cite{Nitinawarat2013} that there exists a closed-loop method that achieves a better detection-error exponent than an open-loop scheme. However, the benefit of Theorem~\ref{thm:2} is in improving the communication rate from the compound channel capacity, i.e.,
$$\max_{P_X\in\mathcal{P}_{\mathcal{X}}}\min_{s\in\calS} \mathbb{I}(P_X,W_{Y|X,s})$$ 
to the worst-case capacity, i.e., $$\min_{s\in\calS}\max_{P_X\in\mathcal{P}_{\mathcal{X}}} \mathbb{I}(P_X,W_{Y|X,s}).$$
Note that the worst-case capacity is exactly the maximal achievable rate described in Theorem~\ref{thm:2} because
\begin{align*}
    \max_{\{P_{X,s''}\}_{s''\in\mathcal{S}}\in\left(\mathcal{P}_{\calX}\right)^{|\mathcal{S}|}} \min_{s\in\calS} \mathbb{I}(P_{X,s},W_{Y|X,s}) = \min_{s\in\calS}\max_{P_X\in\mathcal{P}_{\mathcal{X}}} \mathbb{I}(P_X,W_{Y|X,s}). 
\end{align*}
A similar technique exploiting the feedback to improve the capacity in a compound channel can be found in \cite{Shrader}. 
 
\begin{remark}
    The regions characterized in Theorem~\ref{thm:2} and Theorem 3 are identical when the condition in Corollary 5 is satisfied and the worst-case capacity is equal to the compound channel capacity.
\end{remark}

\subsection{Bi-Static Model}
\label{sec:bi-static-result}
In joint communication and sensing, codewords convey information and induce a codeword-dependent distribution from which observations are generated and are used to estimate the state. In the mono-static model, state estimation is performed at the transmitter, which has the access to the transmitted codeword, so that a likelihood-based state estimation can be performed with exact knowledge of the transmitted codeword. In the bi-static model, however, the transmitted codeword is unknown to the state estimator $g^{(b)}$.

A natural method to estimate the state is to use a successive approach, by which one first estimates the codeword and then performs a maximum likelihood estimation given the decoded codeword. Denoting by $\mathbf{C}^{(b)}$ the closure of all achievable regions in the bi-static model, our next theorem provides an inner bound of $\mathbf{C}^{(b)}$ based on successive schemes. 

\begin{theorem}
The closure region of all achievable joint communication and sensing rate/detection error exponent pairs for the bi-static joint communication and sensing model satisfies

\begin{align}
\mathbf{C}^{(b)} &\supseteq \mathcal{D}_{\textnormal{succ}} \triangleq \bigcup_{P_X\in\mathcal{P}_{\calX}}\left\{\begin{array}{l}(R,E)\in\bbR_+^2:\\
R\leq \min_{s\in\calS} \mathbb{I}(P_X,W_{Y|X,s})\\
E\leq \min\left(\rho_{\textnormal{succ}}(P_X,R), \phi(P_X)\right)
\end{array}\right\},
\end{align}

where 
\begin{align}
    \rho_{\textnormal{succ}}(P_X, R) \triangleq \min_{s\in\calS}\min_{\widehat{P}\in\calP_{\calY|\calX}}\left( \mathbb{D}(\widehat{P}\Vert W_{Y|X,s}|P_X)+\left|\mathbb{I}(P_X;\widehat{P})-R\right|^{+}\right). 
\end{align}
\label{thm:bi-static-result}
\end{theorem}
\begin{proof}
Theorem~\ref{thm:bi-static-result} is a direct result of the following inequality \begin{align*}
    P_d^{(n)} &= \max_{s\in\calS} \max_{w\in[1;M]}\bbP(g^{(b)}(Y^n)\neq s|W=w,S=s)\\
    &= \max_{s\in\calS} \max_{w\in[1;M]}\bbP(g^{(b)}(Y^n)\neq s|h^{(b)}=w,W=w,S=s)\mathbb{P}(h^{(b)}=w|W=w,S=s)\\
    &\hspace{0.5cm}+ \max_{s\in\calS} \max_{w\in[1;M]}\bbP(g^{(b)}(Y^n)\neq s|h^{(b)}\neq w,W=w,S=s)\mathbb{P}(h^{(b)}\neq w|W=w,S=s)\\
    &\leq \max_{s\in\calS} \max_{w\in[1;M]}\bbP(g^{(b)}(Y^n)\neq s|h^{(b)}=w,W=w,S=s)\nonumber\\
    &\hspace{0.5cm}+ \max_{s\in\calS} \max_{w\in[1;M]}\bbP(h^{(b)}(Y^n)\neq w|W=w,S=s),
\end{align*}
where we have bounded the probability of making a wrong state estimation by $1$ when the codeword is decoded incorrectly. Note that the term 
$$
    \max_{s\in\calS} \max_{w\in[1;M]}\bbP(h^{(b)}(Y^n)\neq w|W=w,S=s)
$$
is the communication-error probability; it has been shown in \cite[Theorem 10.2]{Csiszar2011} that, for all $P_X$, there exists a constant composition code with type $P_X$ such that the communication-error exponent is lower bounded by $\rho_{\text{succ}}(P_X)$. Moreover, when the codeword is correctly decoded, the exponent of 
$$
    \max_{s\in\calS} \max_{w\in[1;M]}\bbP(g^{(b)}(Y^n)\neq s|h^{(b)}=w,W=w,S=s)
$$
can be lower bounded by the Chernoff information $\phi(P_X)$ as shown in Lemma~\ref{lem:1} in Section~\ref{sec:proofs} and proofs therein. 
\end{proof}

Unfortunately, the loss induced by upper bounding the term $\bbP(g^{(b)}(Y^n)\neq s|h^{(b)}\neq w,W=w,S=s)$ by $1$ might result in a loose  bound on the state exponent.
The difficulty of analyzing the term $\bbP(g^{(b)}(Y^n)\neq s|h^{(b)}\neq w,W=w,S=s)$ comes from the fact that the receiver does not know the conditional type $\widehat{p}_{\bfy|\bfx_w}$. 
By leveraging random constant composition codes and joint decoding/detection, the following result offers improvements.

\begin{theorem}
The closure region of all achievable joint communication and sensing rate/detection error exponent pairs for the bi-static joint communication and sensing model satisfies
\begin{align}
\mathbf{C}^{(b)} &\supseteq \mathcal{D}_{\textnormal{joint}} \triangleq \bigcup_{P_X\in\mathcal{P}_{\calX}}\left\{\begin{array}{l}(R,E)\in\bbR_+^2:\\
R\leq \min_{s\in\calS} \mathbb{I}(P_X,W_{Y|X,s})\\
E\leq \min\left(\rho_{\textnormal{joint}}(P_X, R), \phi(P_X)\right)
\end{array}\right\},
\end{align}
where 
\begin{align}
    \rho_{\textnormal{joint}}(P_X, R) \triangleq \min_{s\in\calS}\min_{\widehat{P}\in\calP_{\calY|\calX}}\left( \mathbb{D}(\widehat{P}\Vert W_{Y|X,s}|P_X)+\left|\min_{s'\neq s}\min_{P'\in\calP_{s,s'}(\widehat{P},P_X,R)}\mathbb{I}(P_X,P')-R\right|^{+}\right), 
    \label{eqn:rho_2}
\end{align}
and 
\begin{align*}
    &\calP_{s,s'}(\widehat{P},P_X,R) \triangleq \Bigg\{P'\in\calP_{\calY|\calX}: 
    \mathbb{D}(P'\Vert W_{Y|X,s'}|P_X) + \bbH(P'|P_X) \nonumber\\
    &\hspace{1cm}\leq \min\left(\beta(\widehat{P},P_X,R,s),\mathbb{D}(\widehat{P}\Vert W_{Y|X,s}|P_X)+\bbH(\widehat{P}|P_X)\right),
    P_X \circ P'  = P_X\circ \widehat{P}\Bigg\},
\end{align*}
where 
\begin{align}
    \beta(\widehat{P},P_X,R,s) \triangleq   \min_{\substack{P''\in\calP_{\calY|\calX}:\mathbb{I}(P_X,P'')<R,\\P_X\circ P''=P_X\circ \widehat{P}}} \bbD(P''\Vert W_{Y|X,s}|P_X) + \bbH(P''|P_X).
\end{align}
\label{thm:bi-static-result2}
\end{theorem}
The calculation of $\rho_{\text{joint}}(P_X, R)$ in (\ref{eqn:rho_2}) involves a minimization over $\widehat{P}\in\calP_{\calY|\calX}$ and $P' \in\calP_{s,s'}(\widehat{P},P_X,R)$. In Corollary~\ref{Cor:3}, we provide a lower bound on $\rho_{\text{joint}}(P_X, R)$ to simplify the expression and show that it is greater than or equal to the compound channel communication-error exponent $\rho_{\text{succ}}(P_X, R)$.

\begin{corollary}
For all $P_X$, it holds that $\rho_{\textnormal{joint}}(P_X, R) \geq \rho_{\textnormal{succ}}(P_X, R)
$, and hence, $ \mathcal{D}_{\textnormal{joint}} \supseteq \mathcal{D}_{\textnormal{succ}}$. Moreover, $\rho_{\textnormal{joint}}(P_X, R)$ is lower bounded by 
\begin{align}
    \rho_{\textnormal{joint}}(P_X,R) \geq \min_s\min\left(\min_{\widehat{P}\in\calP_{\calY|\calX}'(P_X,R,s)}\gamma_1(\widehat{P}, R,s),\min_{\widehat{P}\in\calP_{\calY|\calX}''(P_X,R,s)}\gamma_2(\widehat{P}, R,s)\right),
\end{align}
where for all $\widehat{P}$ and $s$, 
\begin{align*}
    \gamma_1(\widehat{P},R,s) &\triangleq  \mathbb{D}(\widehat{P}\Vert W_{Y|X,s}|P_X)+\left|\mathbb{I}(P_X,\widehat{P})-R\right|^{+}\\
    \gamma_2(\widehat{P},R,s) &\triangleq \mathbb{D}(\widehat{P}\Vert W_{Y|X,s}|P_X) +\Bigg|\mathbb{D}(\widehat{P}\Vert W_{Y|X,s}|P_X) + \mathbb{H}(P_X\circ \widehat{P}) - \beta(\widehat{P},P_X,R,s)-R\Bigg|^{+},
\end{align*}
and 
\begin{align*}
    \calP_{\calY|\calX}'(P_X,R,s) &\triangleq \{\widehat{P}\in\calP_{\calY|\calX}:\max_{s''\in\calS}\beta(\widehat{P},P_X,R,s'') \geq \mathbb{D}(\widehat{P}\Vert W_{Y|X,s}|P_X)+\mathbb{H}(\widehat{P}|P_X)\}\\
    \calP_{\calY|\calX}''(P_X,R,s) &\triangleq \{\widehat{P}\in\calP_{\calY|\calX}:\max_{s''\in\calS}\beta(\widehat{P},P_X,R,s'') < \mathbb{D}(\widehat{P}\Vert W_{Y|X,s}|P_X)+\mathbb{H}(\widehat{P}|P_X)\}.
\end{align*}
    \label{Cor:3}
\end{corollary}
Note that $\rho_{\text{succ}}(P_X,R)$ can be expressed as minimizing $\gamma_1(\widehat{P},R,s)$ over all possible $\widehat{P}$ and all $s$. Therefore, we can show that $\rho_{\textnormal{joint}}$ is larger than $\rho_{\textnormal{succ}}$ if $\gamma_2(\widehat{P},R,s)\geq \gamma_1(\widehat{P},R,s)$ for all $\widehat{P}\in \calP_{\calY|\calX}''(P_X,R,s)$.
When $\widehat{P}\in \calP_{\calY|\calX}''(P_X,R,s)$, one can observe that 
$$
\mathbb{D}(\widehat{P}\Vert W_{Y|X,s}|P_X) + \mathbb{H}(P_X\circ\widehat{P}) - \mathbb{I}(P_X,\widehat{P}) =   \mathbb{D}(\widehat{P}\Vert W_{Y|X,s}|P_X) + \mathbb{H}(\widehat{P}|P_X)
$$ is greater than $\beta(P_X,R,s)$
and hence, $\gamma_2(\widehat{P},R,s)\geq \gamma_1(\widehat{P},R,s)$. Therefore,
\begin{align}
    \rho_{\text{joint}}(P_X, R) &\geq \min_s\min\left(\min_{\widehat{P}\in\calP_{\calY|\calX}'(P_X,s)}\gamma_1(\widehat{P}, R,s),\min_{\widehat{P}\in\calP_{\calY|\calX}''(P_X,s)}\gamma_1(\widehat{P}, R,s)\right)\\
    &= \rho_{\text{succ}}(P_X, R). 
\end{align}

\begin{remark}
  It is known that the communication-error exponent is zero when one transmits at the rate of compound channel capacity. Recall that $\rho_{\textnormal{succ}}$ is calculated by upper bounding the detection-error probability by $1$ when the message decoding error happens. Therefore, $\rho_{\textnormal{succ}}$ is also zero when $R=\min_{s\in\calS}\avgI{P_X,W_{Y|X,s}}$. In contrast, $\rho_{\textnormal{joint}}$ can be positive even when $R=\min_{s\in\calS}\avgI{P_X,W_{Y|X,s}}$. This can be seen from the expression of $\gamma_2(\widehat{P},R,s)$ in Corollary~\ref{Cor:3}, and we will illustrate this by an example given in the next sub-section.
\end{remark}

\subsection{Numerical Examples}
\label{sec:example}

\paragraph{Mono-Static Model} We first consider the channel $W_{YZ|X,S}$ defined in Table~\ref{tab:ex1}. In this example, $\calY=\calZ=\calX=\{0,1\}$ and $W_{Z|X,S}=W_{Y|X,S}$. Note that transmitting the symbol $X=1$ is most useful to identify the channel state  because this is the situation in which the output distributions corresponding to different states are most distinguishable. However, if the transmitter only transmits $X=1$, the communication rate would become zero as seen in Fig.~\ref{fig:cap_1_a}. Therefore, the tradeoff between maximizing the detection-error exponent and maximizing the communication rate can be clearly seen in Fig.~\ref{fig:cap_1_a}. On the other hand, the difference between $\mathbf{C}_{\textnormal{open}}^{(m)}$ and the inner bound of $\mathbf{C}_{\textnormal{closed}}^{(m)}$ given in Theorem~\ref{thm:2} is shown in Fig.~\ref{fig:cap_1_b}, illustrating how the communication rate is increased. Indeed, the inner bound region for $\mathbf{C}^{(m)}_{\textnormal{closed}}$ is larger than $\bfC^{(m)}_{\textnormal{open}}$ because the compound capacity is here strictly less than the worst-case capacity.

In contrast to the channel given in Table \ref{tab:ex1}, for which the tradeoff between sensing and communication exists, the channel given in Table~\ref{tab:ex2} is a \ac{BSC} for which, according to Corollary~\ref{cor:1}, the best error exponent of detection is always achieved regardless of the type of codewords. The result of Theorem~\ref{thm:mono-capacity} for this channel is illustrated in Fig.~\ref{fig:cap_2}.

\begin{table}[h]
\centering
\caption{Table for $W_{Z|X,S}(0) = W_{Y|X,S}(0)$ for all $X\in\{0,1\}$ and $S\in\{0,1,2\}$.}
\begin{tabular}{|c|c|c|}
    \hline
     \backslashbox{$S$}{$X$}
& 0 & 1  \\
     \hline
     0 & 0.95 & 0.45\\
     \hline
     1 & 0.9 & 0.2\\
     \hline
     2 & 0.5 & 0.03\\
     \hline
\end{tabular}
\label{tab:ex1}

\end{table}

\begin{figure}[!h]
    \begin{subfigure}{.47\textwidth}
    \centering
    \includegraphics[scale=0.6]
    {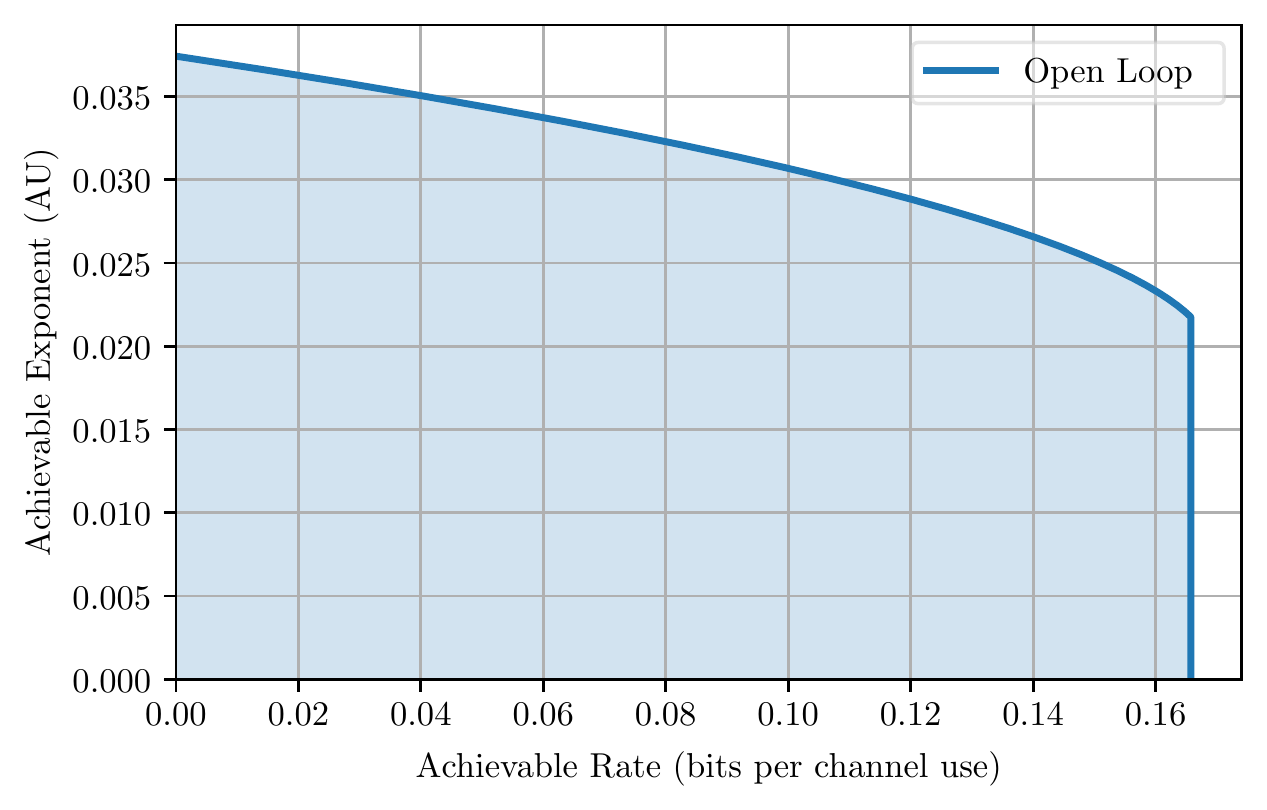}
    \caption{$\mathbf{C}_{\textnormal{open}}^{(m)}$}
    \label{fig:cap_1_a}
    \end{subfigure}
    \begin{subfigure}{.47\textwidth}
    \centering
    \includegraphics[scale=0.6]{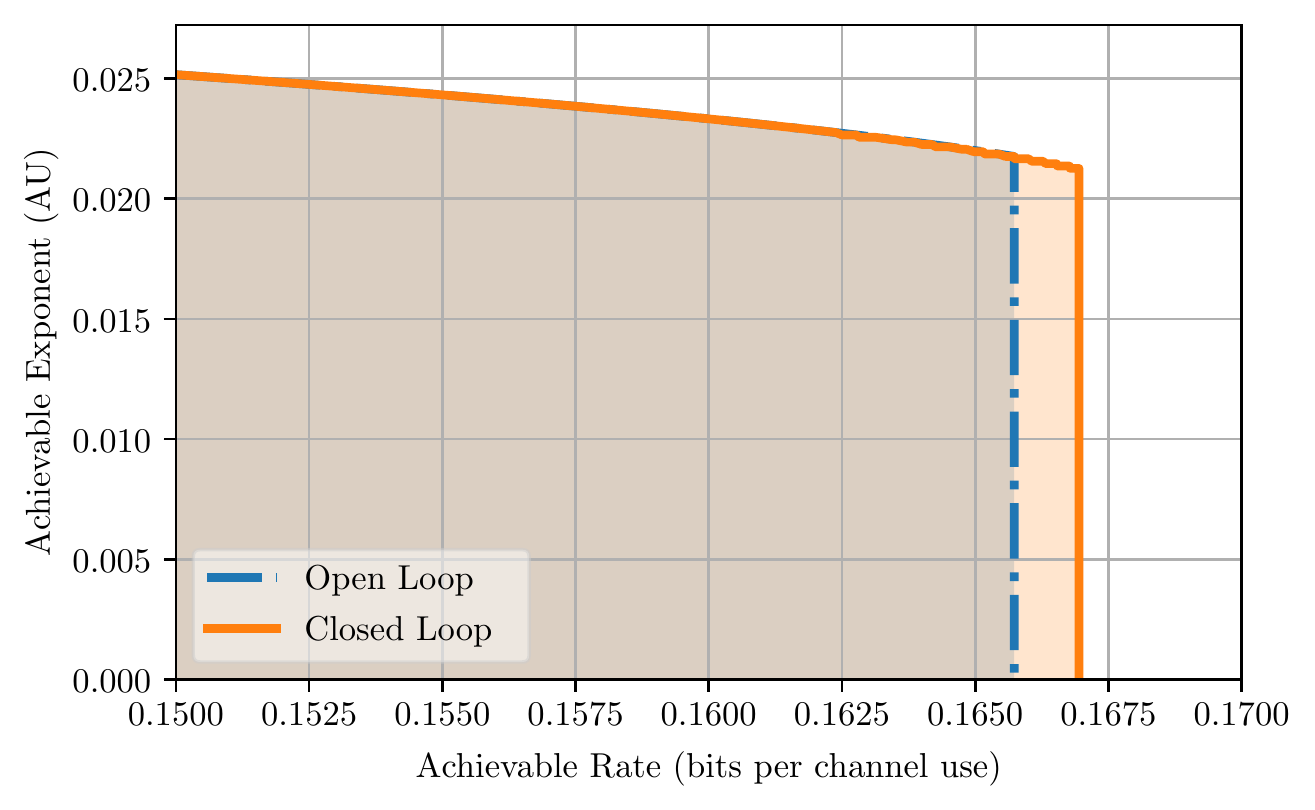}
    \caption{$\mathbf{C}_{\textnormal{open}}^{(m)}$ and the inner bound of $\mathbf{C}_{\textnormal{closed}}^{(m)}$ }
    \label{fig:cap_1_b}
    \end{subfigure}
    \caption{Closure region $\mathbf{C}^{(m)}_{\text{open}}$ and inner bounds for  $\mathbf{C}^{(m)}_{\textnormal{closed}}$ corresponding to the channel of Table~\ref{tab:ex1}.}
    \label{fig:cap_1}
\end{figure}

\begin{table}[!h]
\centering
\caption{Table for $W_{Z|X,S}(0) = W_{Y|X,S}(0)$ for all $X\in\{0,1\}$ and $S\in\{0,1,2\}$.}
\begin{tabular}{|c|c|c|}
    \hline
     \backslashbox{$S$}{$X$}
& 0 & 1  \\
     \hline
     0 & 0.9 & 0.1\\
     \hline
     1 & 0.8 & 0.2\\
     \hline
     2 & 0.7 & 0.3\\
     \hline
\end{tabular}
\label{tab:ex2}

\end{table}

\begin{figure}[h]
    \centering
    \includegraphics[scale=0.65]{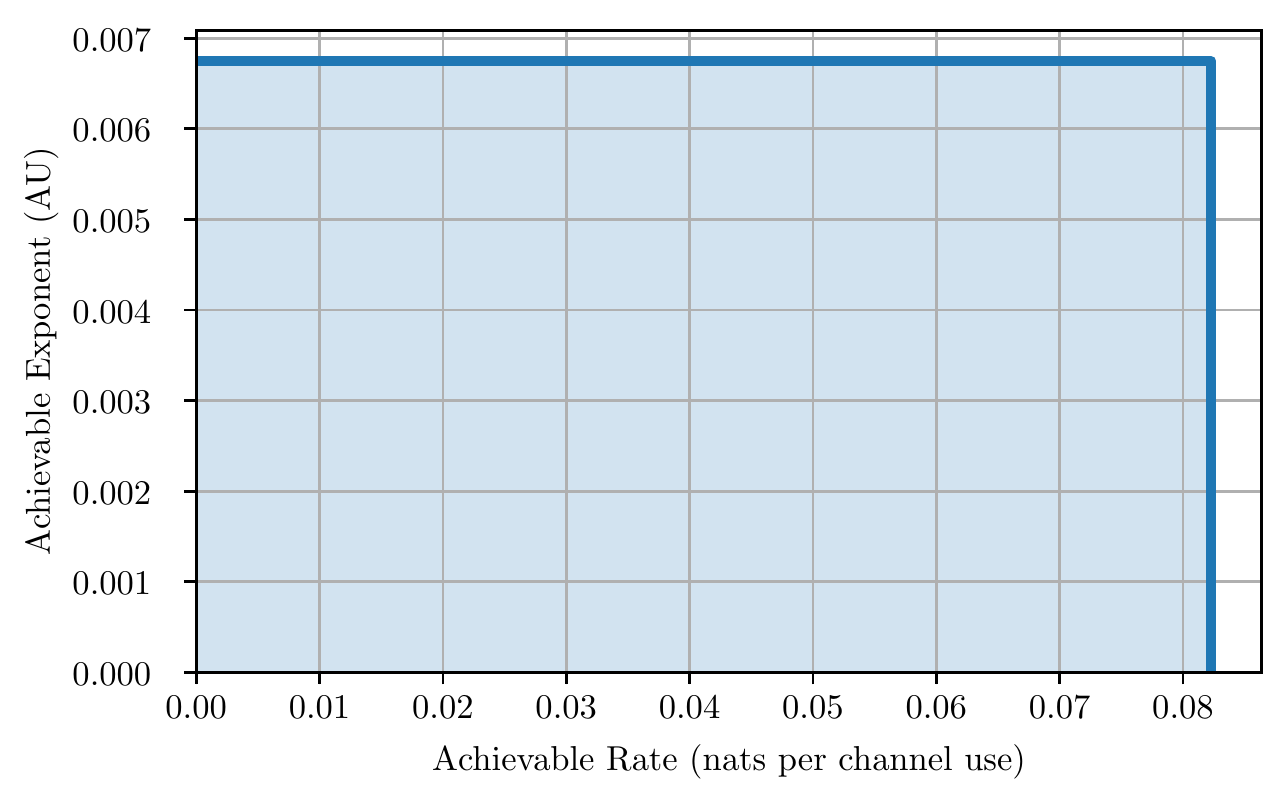}
    \caption{Closure region $\mathbf{C}^{(m)}_{\text{open}}$ corresponding to the channel of Table~\ref{tab:ex2}.}
    \label{fig:cap_2}
\end{figure}

% \begin{figure}[h]
%     \centering
%     \includegraphics[scale=0.65]{Figures/Cap_region_binary.pdf}
%     \caption{Capacity region $\mathbf{C}^{(m)}_{\text{open}}$ corresponding to binary memoryless channel in~\cite[Sec. III-A]{Joudeh2022Joint}.}
%     \label{fig:cap_bin}
% \end{figure}
% \begin{figure}[h]
%     \centering
%     \includegraphics[scale=0.65]{Figures/Cap_region_gaussian.pdf}
%     \caption{Capacity region $\mathbf{C}^{(m)}_{\text{open}}$ corresponding to Gaussian channel in~\cite[Sec. III-B]{Joudeh2022Joint}.}
%     \label{fig:cap_gaussian}
% \end{figure}
\paragraph{Bi-Static Model}

For the bi-static model, we present a numerical example as defined in Table~\ref{tab:ex3}. Here, $\calY=\calZ=\calX=\{0,1\}$ and $\calS=\{0,1\}$. When $\calS=0$, the channel $W_{Y|X,0}$ is a \ac{BSC} with cross-over probability $0.3$, and $W_{Y|X,1}$ is a \ac{BSC} with cross-over probability $0.6$.  From  Fig.~\ref{fig:cap_bistatic_1}, one observes that $\calD_{\text{joint}}$ is strictly larger than $\calD_{\text{succ}}$ especially when the code rate is high.
When the rate is low, the exponent $\rho_{\text{joint}}$ is dominated by the Chernoff information $\phi(P_X)$, which is approximately $0.048$ in this example. When the rate is high, the exponent is dominated by a rate-dependent term, capturing the fact the receiver's inability to easily decode the transmitted sequence is what drives performance. 
By the definition of $\rho_{\text{succ}}$, the value of $\rho_{\text{succ}}$ is zero when the code rate is at the compound channel capacity. However, the joint detection scheme reaches a positive exponent even when transmitting at the compound capacity, highlighting the benefits of joint detection over successive decoding.

\begin{table}[!h]
\centering
\caption{Table for $W_{Y|X,S}(0)$ for all $X\in\{0,1\}$ and $S\in\{0,1\}$.}
\begin{tabular}{|c|c|c|}
    \hline
     \backslashbox{$S$}{$X$}
& 0 & 1  \\
     \hline
     0 & 0.7 & 0.3\\
     \hline
     1 & 0.4 & 0.6\\
     \hline
\end{tabular}
\label{tab:ex3}
\end{table}

\begin{figure}[h]
    \centering
    \includegraphics[scale=0.65]{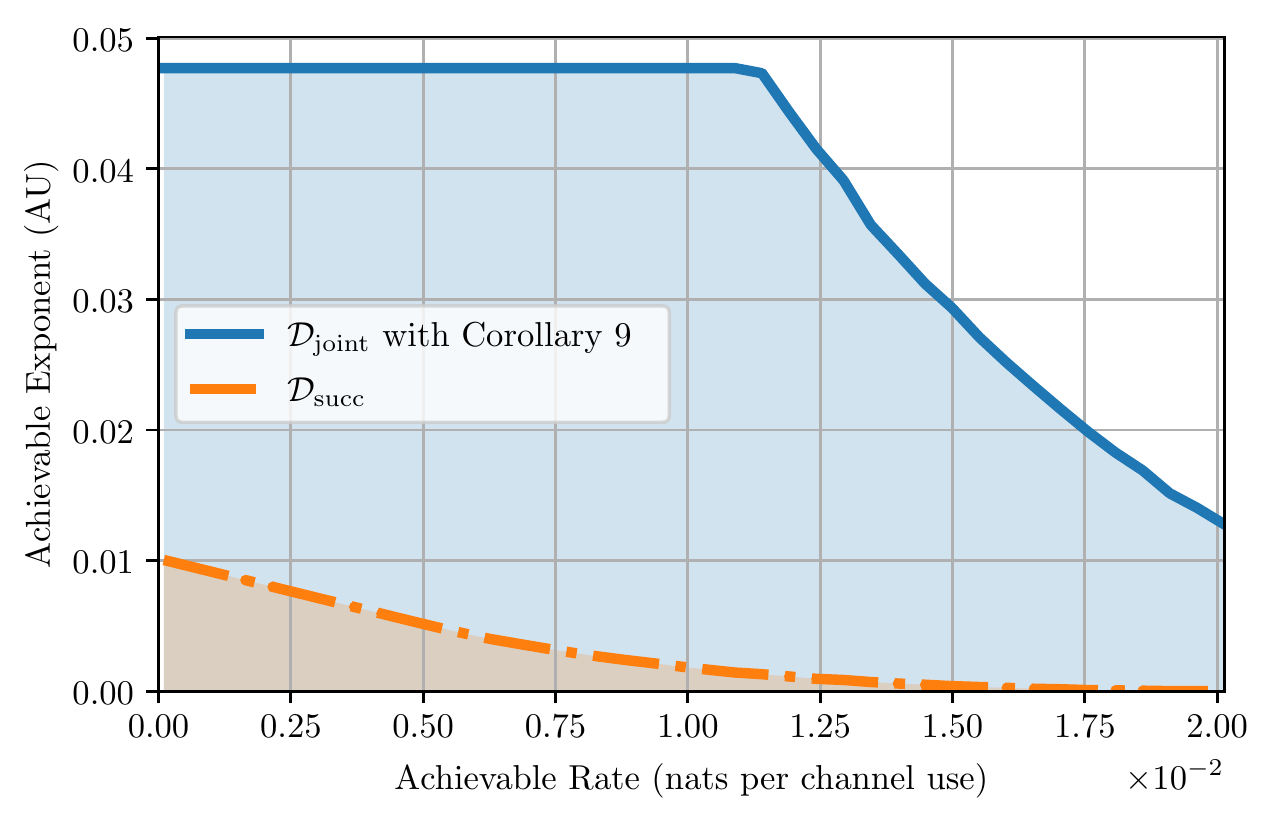}
    \caption{Region $\calD_{\text{succ}}$ and $\calD_{\text{joint}}$ corresponding to Table~\ref{tab:ex3}.}
    \label{fig:cap_bistatic_1}
\end{figure}
\FloatBarrier

\section{Detailed Proofs}
\label{sec:proofs}
\subsection{Achievability Proof of  Theorem~\ref{thm:mono-capacity}}
\label{sec:achievability}
We show that all $(R,E)$ pairs within the region $\mathbf{C}_{\text{open}}^{(m)}$ are achievable. Since we restrict ourselves to open-loop schemes, we may fix $P_X$ as the type of all codewords. Fix any $\epsilon>0$. By~\cite[Theorem 10.2]{Csiszar2011}, there exists a code with encoder ${f}^{(m)}$ and decoder $h^{(m)}$ such that ${f}^{(m)}(w)\in \mathcal{T}_{P_X}^n$ for all $w$, the rate is at least $\min_{s\in\calS}\mathbb{I}(P_X,W_{Y|X,s})-\epsilon$, and $\max_w\mathbb{P}(h(Y^n)\neq w|S=s)<\epsilon$ for all $s\in\mathcal{S}$. Then, the detection-error exponent $\phi(P_X)$ is given by Lemma~\ref{lem:1} adapted from \cite[Theorem 1]{Nitinawarat2013}. 
\begin{lemma}
Suppose that the the codeword corresponding to the message $w\in[1;M]$ has type $P_X\in \mathcal{P}_{\calX}$, the conditional detection-error exponent  $E_{\textnormal{d},w}\triangleq -\frac{1}{n}\log \max_{s\in\calS}\mathbb{P}(g(Z^n)\neq s|S=s,W=w)$ in an open-loop scheme is asymptotically upper bounded by
\begin{align*}
  \phi(P_X)&\eqdef \min_{s} \min_{s'\neq s}\max_{\ell\in [0,1]} -\sum_{x}P_X(x)\log \left(\sum_{z} W_{Z|X,s}(z|x)^{\ell}W_{Z|X,s'}(z|x)^{1-\ell}\right). 
  \end{align*}
  \label{lem:1}
Moreover, it is also asymptotically achievable by a maximum likelihood estimator $g_{\textnormal{ML}}$. Specifically, 
\begin{align}
    \max_{s\in\calS}\mathbb{P}\left(g_{\textnormal{ML}}(Z^n)\neq s\middle|S=s,W=w\right) \leq \Theta_n(1) e^{-\phi(P_X)}
\end{align}
for all $w\in[1;M]$, and for all $s\in\calS$
\begin{align}
    \mathbb{P}\left(g_{\textnormal{ML}}(Z^n)\neq s\middle|S=s,W=w\right) \leq \Theta_n(1) e^{-\psi_s(P_X)},
\end{align}
for all $w\in[1;M]$, where 
\begin{align}
    \psi_s(P_X)&\eqdef \min_{s'\neq s}\max_{\ell\in [0,1]} -\sum_{x}P_X(x)\log \left(\sum_{z} W_{Z|X,s}(z|x)^{\ell}W_{Z|X,s'}(z|x)^{1-\ell}\right). 
\end{align}
\end{lemma}
Taking the union over all possible $P_X$ leads to the region in~\eqref{eqn:thm1_1}.
% Next, we show that all the $(R,E)$ tuples described by (\ref{eqn:thm1_1}) is achievable. All the message rate smaller than $\max_{P_X\in\mathcal{P}_X} \min_{s\in\calS}\mathbb{I}(W_{Y|X,s},P_X)$ is achievable by the result of compound channel capacity when $n$ is sufficiently large. Now, we fix some $$R\leq \max_{P_X\in\mathcal{P}_X} \min_{s\in\calS}\mathbb{I}(W_{Y|X,s},P_X),$$ it is known that there exist a code with type $P_X$ and rate $\frac{1}{n}\log M\geq R-4\epsilon''$ and error probability $\mathbb{P}(h(Y^n)\neq W)\leq \epsilon$ for any $\epsilon,\epsilon''>0$ whenever $\min_{s\in\calS}\mathbb{I}(W_{Y|X,s},P_X)\geq R-\epsilon'$. Therefore, we have the detection-error exponent $E_{\text{D}}$ smaller than $$ \max_{P_X\in \{P_X\in\mathcal{P}_X:\min_{s\in\calS}\mathbb{I}(W_{Y|X,s},P_X)\geq R-\epsilon'\}} \phi(P_X)$$
% is achievable. Since the function $\phi$ is continuous in $P_X$, we conclude that the detection-error exponent $\max_{P_X\in \mathcal{T}(R)}\phi(P_X)$ is achievable. 
 
\subsection{Converse Proof of Theorem~\ref{thm:mono-capacity}}
\label{sec:converse}
%\cmnt{[REWRITE AS ITW VERSION]}
%$R\leq\max_{P_X}\min_s\I{W_{Y|X,s}}{P_X}$ is by the compound channel capacity. 
Assume that the rate/detection-error exponent pair $(R,E)$ is achievable. Then, for all $\epsilon>0$, there exists $n$ sufficiently large and a codebook $\calC^{(m)}$ such that
\begin{align*}
    \frac{\log\card{\calC^{(m)}}}{n}&\geq R-\epsilon \\
    \max_{s\in\calS}\max_{w\in[1;M]}\mathbb{P}(h(Y^n)\neq w|W=w,S=s) &\leq \epsilon\\
    -\frac{1}{n}\log \max_{s\in\calS}\max_{w\in\intseq{1}{M}} \mathbb{P}(g(Z^n)\neq s|W=w,S=s)&\geq E-\epsilon.
\end{align*}
Since there is at most a polynomial number of types, there exists a set of types $\calT$ such that, for all $P_X\in\calT$,
the subcode $\calC_{P_X}\triangleq \{{f}(w):\hat{p}_{{f}(w)}=P_X\} \subset \calC^{(m)}$ satisfies
\begin{align}
    \max_{s\in\calS}\max_{w\in f^{-1}(\calC_{P_X})}\mathbb{P}(h(Y^n)\neq w|W=w,S=s) &\leq \epsilon
\end{align}
and
\begin{align}
    \frac{\log\card{\calC_{P_X}}}{n}>\frac{\log\card{\calC^{(m)}}}{n}-\delta\geq R-\epsilon-\delta \label{eq:lb-codebook-size}
\end{align}
for some $\delta$ vanishing with $\epsilon$.
Fix any $P_X\in\calT$. Let $$\widetilde{P}^n_X(\bfx)\eqdef\frac{1}{\card{\calC_{P_X}}}\sum_{\tilde{\bfx}\in\calC_{P_X}}\indic{\bfx=\tilde{\bfx}}.$$ Observe that
\begin{align}
\label{eq:avg-type}
  \overline{P}_X(x)\eqdef\frac{1}{n}\sum_{i=1}^n\widetilde{P}_{X_i}(x)=P_X(x),
\end{align}
where $\widetilde{P}_{X_i}$ is the $i$-th marginal distribution of $\widetilde{P}^n_X$. 
Then, for any state $s$,
\begin{align}
  (1-\epsilon)\log \card{\calC_{P_X}}-\Hb{\epsilon}&\labrel\leq{eq:Fano}\avgI{\widetilde{P}_X^n,W^{\otimes n}_{Y|X,s}}\nonumber\\
                        &\labrel={eq:ensemble-density}\avgH{Y^n}-\sum_{\bfx\in\calX^n}\widetilde{P}^n_X(\bfx)\avgH{Y^n|X^n=x^n,s}\nonumber\\
                        &=\avgH{Y^n}-\sum_{\bfx\in\calX^n}\widetilde{P}^n_X(x^n)\sum_{i=1}^n\avgH{Y_i|X^n=x^n,Y^{i-1},s}\nonumber\\
  &=\sum_{i=1}^n\avgH{Y_i|Y^{i-1}}-\sum_{i=1}^n\sum_{x_i\in\calX}\widetilde{P}_{X_i}(x_i)\avgH{Y_i|X_i=x_i, s}\nonumber\\
&\leq\sum_{i=1}^n\avgI{\widetilde{P}_{X_i}, W_{Y|X,s}}\nonumber\\
&\labrel\leq{eq:concavity} n\avgI{\overline{P}_X, W_{Y|X,s}}
\end{align}
where $\Hb{\cdot}$ is the binary entropy function, \eqref{eq:Fano} follows from Fano's inequality and standard techniques, \eqref{eq:ensemble-density} follows by identifying $Y^n\sim\widetilde{P}_X^n\circ W^{\otimes n}_{Y|X,s}$ and the definition of mutual information, and \eqref{eq:concavity} follows from the concavity of mutual information in the input distribution. Then, we obtain
\begin{align}
\label{eq:conv-sing-letter}
  \frac{\log \card{\calC_{P_X}}}{n}\leq\frac{\avgI{P_X, W_{Y|X,s}}+\frac{1}{n}}{1-\epsilon},
\end{align}
where we have used \eqref{eq:avg-type}.

%\textcolor{blue}{What is the role of $\calT$ in this derivation?}
Since \eqref{eq:conv-sing-letter} is valid for any $\epsilon$ and state $s$, the size of sub-codebook $\calC_{P_X}$ is upper bounded by the mutual information in a compound channel, i.e.,
\begin{align}
  \frac{\log \card{\calC_{P_X}}}{n}<\min_{s\in\calS}\mathbb{I}(P_X,W_{Y|X,s})+\tau \label{eq:boundrate}
\end{align}
for some $\tau>0$ vanishing with $\epsilon$.
%By \cite[Corollary 6.4]{CodingTheoremsDMC2},
%\begin{equation}
%    \frac{\log\card{\calC_{P_X}}}{n}<\min_{s}\mathbb{I}(P_X,W_{Y|X,s})+\tau \label{eq:boundrate}
%\end{equation}
%for some $\tau$ vanishing with $\epsilon$. Choose now $P_X^*\in\calT$ such that
%\begin{align}
%    P_X^* \eqdef\argmin_{P_X\in\calT}\phi(P_X)
%\end{align}
%with $\phi(\cdot)$ defined in~\eqref{eq:phi}.
On the other hand, for this $P_X\in\calT$,
\begin{align}
    &E-\epsilon\nonumber\\*
    &\leq-
    \frac{1}{n}\log\max_{s\in\calS}\max_{w\in\intseq{1}{M}} \mathbb{P}(g(Z^n)\neq s|W=w,S=s)\nonumber\\*
    &\stackrel{(a)}{\leq}-
    \frac{1}{n}\log \max_{s\in\calS}\max_{w\in f^{-1}(\calC_{P_X})} \mathbb{P}(g(Z^n)\neq s|W=w,S=s)\nonumber\\
    %&\stackrel{(b)}{=}-
    %\frac{1}{n}\log \max_{s\in\calS}\mathbb{P}(g(Z^n)\neq s|S=s,W=w) \nonumber\\
    &\stackrel{(b)}{\leq} \phi(P_X) + \delta,\label{eq:boundexponent} 
\end{align}
where $(a)$ follows by restricting the set to only the terms corresponding to messages in $\calC_{P_X}$; $(b)$ follows since the detection error is upper bounded by $\phi(P_X)$ for any message with type $P_X$ by Lemma~\ref{lem:1}. Combining~\eqref{eq:boundrate} and~\eqref{eq:boundexponent} and choosing $P_X^*\eqdef\argmin_{P_X\in\calT}\phi(P_X)$, we conclude that for all $\epsilon>0$, there exist $\tau, \delta>0$ vanishing with $\epsilon$ such that
\begin{align}
  R&\leq \min_{s}\mathbb{I}(P_X^*,W_{Y|X,s})+\tau +\epsilon+\delta\\
  E &\leq \phi(P_X^*) + \delta + \epsilon.
\end{align}
Since $\epsilon,\tau,\delta$ can be chosen arbitrarily small as the block length $n$ goes to infinity, %taking the union over all $P_X^*$ yields the desired converse result.
% \begin{align}
%     E_{\text{d}} &= \lim_{n\to\infty}-\frac{1}{n}\log\frac{1}{\card{C}}\sum_{w}\max_s \P{h(Y^n)\neq s|S=s,W=w}\\
%       &\leq \lim_{n\to\infty}-\frac{1}{n}\log\frac{1}{\card{C}}\sum_{\{w:\mathbf{f}(w)\in \mathcal{C}_{P_X}\}}\max_s\P{h(Y^n)\neq s|S=s,W=w}\\
%       &\leq \delta + \phi(P_X),
%       \label{eqn:prof_thm1_2}
% \end{align}
% where (\ref{eqn:prof_thm1_2}) is from Lemma~\ref{lem:1}. By choosing $\delta$ arbitrarily small, we conclude
% \begin{equation}
%     E_{\textnormal{D}}\leq \phi(P_X). 
% \end{equation}
$E$ is upper bounded by $\phi(P_X)$ for some $P_X \in\mathcal{P}_{\calX}$ and the rate $R$ is achieved by  this $P_X$. Taking the union over all possible $P_X$ completes the result of converse of Theorem~\ref{thm:mono-capacity}.

\subsection{Proof of Theorem~\ref{thm:2}}

%To simplify notation in this section alone, we write $W_{Z|x,s}(z)$ instead of $W_{Z|X,s}(z|x)$.

\subsubsection{Definition of the Code}
Before the protocol starts, we fix $n\in\bbN^*$, $\Delta_1,\Delta_2\in(0, 1)$ such that $\Delta_1+\Delta_2<1$, and we also fix some $|\mathcal{S}|$-tuple $\{P_{X,s''}\}_{s''\in\mathcal{S}}\in \left(\mathcal{P}_X^{(1-\Delta_1-\Delta_2)n}\right)^{ |\mathcal{S}|}$, where $\left(\mathcal{P}_X^{(1-\Delta_1-\Delta_2)n}\right)^{|\mathcal{S}|}$ is the set of tuples of types in $\mathcal{P}_{\mathcal{X}}^{(1-\Delta_1-\Delta_2)n},$
as well as the number of messages $M$ as
\begin{align}
    M = n(1-\Delta_1-\Delta_2) \times \left(\min_{s''\in\mathcal{S}} \mathbb{I}(P_{X,s''}, W_{Y|X,s''})-\delta\right)
\end{align}
for some $\delta>0$. 

% We show that the tuple $(R,E)$ is $s$-achievable whenever $R\leq \mathbb{I}(P_X,W_{Y|X,s})$ and $E\leq \psi_s(P_X)$. We start by defining the code $\calC^{(m)}=(\{f_i\}_{i\in [1;n]},g,h)$. 
We first define $g_{\textnormal{ML},i}$ as the maximum likelihood estimator at each time $i\in\mathbb{N}^*$, i.e., 
$$g_{\textnormal{ML},i}(x^{i},z^{i}) \triangleq \argmax_{s\in\calS} \prod_{\ell=1}^{i} W_{Z|X,s}(z_{\ell}|x_{\ell}).$$
Then, the code $\calC^{(m)}=(\{f_i\}_{i\in [1;n]},g^{(m)},h^{(m)})$ is defined through the following steps. 

\paragraph{Initial Estimation} We define $P_{X}^{\#}=\argmax_{P_X\in\mathcal{P}_{\calX}^{\Delta_1n}} \phi(P_X)$ and pick any length $n\Delta_1$ sequence $\mathbf{v}=(v_1,...,v_{\Delta_1 n})$ from the type class $\mathcal{T}_{P_X^{\#}}$. Then, for $1\leq i\leq \Delta_1 n$, the encoder is defined as \begin{align}
    f_i(w,z^{i-1}) = v_i
\end{align}
for all $w\in [1;M]$ and $z^{i-1}\in\calZ^{i-1}$. At time $\Delta_1 n+1$, the transmitter  estimates the state by using the maximum likelihood estimator $\tilde{s} = g_{\textnormal{ML},\Delta_1n }(x^{\Delta_1 n},z^{\Delta_1 n})$.

\paragraph{State Information Transmission} The transmitter then conveys the information of the estimated state to the receiver by encoding the estimated state $\tilde{s}$ into a codeword. Since $|\calS|$ does not grow with $n$, there exists a length $\Delta_2 n$ channel code $(\hat{f},\hat{g})$ with arbitrarily small error probability, where $\hat{f}:\calS\mapsto \calX^{\Delta_2 n}$ is the encoder and $\hat{g}:\calY^{\Delta_2 n}\mapsto \calS$ is the decoder. Denoting $\hat{\mathbf{x}}(\tilde{s}) = (\hat{x}_1(\tilde{s}),...,\hat{x}_{\Delta_2 n}(\tilde{s}))=\hat{f}(\tilde{s})$ the codeword corresponding to $\tilde{s}$. Then, for $\Delta_1 n < i \leq (\Delta_1+\Delta_2)n$, the encoder is defined as 
    \begin{align}
        f_i(w,z^{i-1}) = \hat{x}_{i-\Delta_1 n} (g_{\textnormal{ML},\Delta_1n}(x^{\Delta_1 n},z^{\Delta_1 n}))
    \end{align}
for all $w\in[1;M]$ and $z^{i-1}\in \calZ^{i-1}$. 
\paragraph{Message Transmission}
It is known that for every $\Bar{P}_X\in\mathcal{P}_{\calX}$, there exists a channel code with arbitrarily small error probability $\epsilon>0$ such that the rate is at least $\mathbb{I}(\Bar{P}_X,W_{Y|X,s})-2\tau$ for any $\tau>0$ when the channel state is $s\in\mathcal{S}$ and the number of channel uses is large enough.  Therefore, for any $s\in\mathcal{S}$, there exist an $((1-\Delta_1-\Delta_2)n,\epsilon)$ constant composition code with the type $P_{X,s}$ for the state $s$ channel with rate $\min_{s''\in\mathcal{S}} \mathbb{I}(P_{X,s''}, W_{Y|X,s''})-\delta$ for any $\epsilon>0$, where the set of types $\{P_{X,s''}\}_{s''\in\mathcal{S}}\in\left(\mathcal{P}_X^{(1-\Delta_1-\Delta_2)n}\right)^{ |\mathcal{S}|}$ is fixed at the beginning of the protocol. Let the channel code for the channel with state $s$ be characterized by  $(\tilde{f}_s,\tilde{h}_s)$, where  $\tilde{f}_s:[1;M]\mapsto \mathcal{X}^{(1-\Delta_1-\Delta_2)n}$ is the encoder and $\tilde{h}_s:\calY^{(1-\Delta_1-\Delta_2)n}\mapsto [1:M]$ is the decoder. Denoting $\tilde{\mathbf{x}}(w,s)=(\tilde{x}_1(w,s),...,\tilde{x}_{(1-\Delta_1-\Delta_2)n}(w,s)) = \tilde{f}_s(w)$ the codeword corresponding to the message $w\in[1;M]$. Then, for $(\Delta_1 + \Delta_2) n < i\leq n$, we define the encoder as
\begin{align}
    f_i(w,z^{i-1}) = \tilde{x}_{i-(\Delta_1+\Delta_2) n}(w,g_{\textnormal{ML},\Delta_1n}(x^{\Delta_1 n},z^{\Delta_1 n}))
\end{align}
for all $w\in[1;M]$ and $z^{i-1}\in\calZ^{i-1}$. 
\paragraph{Message Decoding and State Estimation}
Finally, the message decoder is 
\begin{align}
    h^{(m)}(y^n) = \tilde{h}_{\hat{g}\left(y_{\Delta_1 n +1}^{\Delta_2 n}\right)}\left(y_{(\Delta_1+\Delta_2) n+1}^n\right)
\end{align}
for all $y^n\in\mathcal{Y}^n$, and the state estimator is 
\begin{align}
    g^{(m)}(x^n,z^n) = \argmax_{s\in\calS} \prod_{\ell=(\Delta_1+\Delta_2)n+1}^{n} W_{Z|X,s}(z_{\ell}|x_{\ell})
    \label{eqn:state_estimator}
\end{align}
for all $x^n\in\mathcal{X}^n$ and $z^n\in\mathcal{Z}^n$. 
Note that we only use the subsequences $x_{(\Delta_1+\Delta_2)n+1}^n$ and $z_{(\Delta_1+\Delta_2)n+1}^n$ for the 
state estimation in \eqref{eqn:state_estimator}.

\subsubsection{Analysis of the Communication Rate and Detection-Error Exponent}
\paragraph{Rate Analysis}
For any $s\in\mathcal{S}$ any $w\in[1;M]$,  the error probability of communication is 
\begin{align}
    &\mathbb{P}(h(Y^n)\neq W|W=w,S=s) \nonumber\\
    &\leq \mathbb{P}(g_{\textnormal{ML},\Delta_1n}(X^{\Delta_1 n},Z^{\Delta_1 n})\neq s|W=s,S=s) + \mathbb{P}(\hat{h}(Y_{\Delta_1 n+1}^{\Delta_2 n})\neq s,\Tilde{s}=s|W=s,S=s) \nonumber\\
    &\phantom{===}+ \mathbb{P}\left(\tilde{h}_s\left(Y_{(\Delta_1+\Delta_2) n+1}^n\right)\neq w,\Tilde{s}=s,\hat{h}(Y_{\Delta_1 n+1}^{\Delta_2 n})= s\middle|W=w,S=s \right),
    \label{eqn:thm6_2_1}
\end{align}
where we have applied the union bound. 
The first term of \eqref{eqn:thm6_2_1} comes from the event in which the initial estimation of the state $\Tilde{s}$ is incorrect; the second term of \eqref{eqn:thm6_2_1} is the event in which the decoding of the initial estimated state $\Tilde{s}$ is incorrect; the last term of \eqref{eqn:thm6_2_1} is the decoding error probability. 
For any $s\in\mathcal{S}$ and $w\in[1;M]$, all terms on the right-hand side of \eqref{eqn:thm6_2_1} are arbitrarily small when $n$ is sufficiently large by our construction of the code. The rate of communication is 
\begin{align*}
    R &= 
    \frac{1}{n}\log e^{n(1-\Delta_1-\Delta_2)\left(\min_{s''\in\mathcal{S}} \mathbb{I}(P_{X,s''}, W_{Y|X,s''})-\delta\right)} = (1-\Delta_1-\Delta_2)\left(\min_{s''\in\mathcal{S}} \mathbb{I}(P_{X,s''}, W_{Y|X,s''})-\delta\right).
\end{align*}
By making $\Delta_1,\Delta_2$ and $\delta$ arbitrarily small, we conclude that $\min_{s''\in\mathcal{S}} \mathbb{I}(P_{X,s''}, W_{Y|X,s''})$ is achievable. 

\paragraph{Detection-Error Exponent Analysis}
The error probability of detection is 
\begin{align}
    P_d^{(n)} &= \max_{s\in\calS}\max_{w\in\intseq{1}{M}} \mathbb{P}(g^{(m)}(X^n,Z^n)\neq s|W=w,S=s) \nonumber\\
    &= \max_{s\in\calS}\max_{w\in\intseq{1}{M}} \mathbb{P}\left(\argmax_{s\in\calS} \prod_{\ell=(\Delta_1+\Delta_2)n+1}^{n} W_{Z|X,s}(z_{\ell}|x_{\ell}) \neq s \middle|W=w,S=s\right). 
    \label{eqn:thm6_2_2}
\end{align}
Note that $\argmax_{s\in\calS} \prod_{\ell=(\Delta_1+\Delta_2)n+1}^{n} W_{Z|X,s}(z_{\ell}|x_{\ell}) \neq s$ is the error event of applying the ML estimator when the type of the input sequence is $\hat{p}_{X_{(1-\Delta_1-\Delta_2)n+1}^n} = P_{X,\Tilde{s}}$. For all $t\in [n(\Delta_1+\Delta_2)+1;n]$, the channel input $X_t$ is chosen without using the feedback $Z_{n(\Delta_1+\Delta_2)+1}^n$. Moreover, for any $w\in[1;M]$, the sequence $X_{n(\Delta_1+\Delta_2)+1}^n$ has a constant type, which depends on $\Tilde{s}$. 
Therefore, we can apply Lemma~\ref{lem:1} and the law of total probability to upper bound the error probability of \eqref{eqn:thm6_2_2} as 
\begin{align}
    P_d^{(n)} &\leq \max_{s\in\mathcal{S}}\max_{w\in [1;M]}\sum_{s''\in\mathcal{S}}\mathbb{P}\left(\Tilde{s}=s''\middle|W=w,S=s\right) 
    \times \Theta_n(1)\times e^{-n(1-\Delta_1-\Delta_2)\psi_s(P_{X,s''})}\\
    &\leq \max_{s\in\mathcal{S}} \Theta_n(1)\times e^{-n(1-\Delta_1-\Delta_2)\min_{s''\in\mathcal{S}}\psi_s(P_{X,s''})}\label{eqn:thm6_2_3}\\
    &= \Theta_n(1)\times e^{-  n(1-\Delta_1-\Delta_2)\min_{s\in\mathcal{S}}\min_{s''\in\mathcal{S}}\psi_s(P_{X,s''})}
    \\
    &= \Theta_n(1)\times e^{- n(1-\Delta_1-\Delta_2)\min_{s''\in\mathcal{S}}\phi(P_{X,s''})},
    \label{eqn:thm6_2_4}
\end{align}
where in \eqref{eqn:thm6_2_3} we lower bound the exponent $\psi_s(P_{X,s''})$ by the minimum one and $\Theta_n(1)$ is some constant, and in \eqref{eqn:thm6_2_4} we swap the order of $\min_{s\in\mathcal{S}}$ and $\min_{s''\in\mathcal{S}}$ and apply the definition of $\phi$. 
Then, we have shown that 
\begin{align}
    E_d^{(n)} \geq (1-\Delta_1-\Delta_2)\min_{s\in\mathcal{S}}\phi(P_{X,s}). 
\end{align}
By taking the union over all possible choices of $\{P_{X,s''}\}_{s''\in\mathcal{S}}\in \left(\mathcal{P}_X^{(1-\Delta_1-\Delta_2)n}\right)^{ |\mathcal{S}|}$ and making $\Delta_1,\Delta_2$ arbitrarily small, we conclude that the following region is achievable. 
\begin{align}
    \bigcup_{\{P_{X,s''}\}_{s''\in\mathcal{S}}\in\left(\mathcal{P}_X^{(1-\Delta_1-\Delta_2)n}\right)^{ |\mathcal{S}|}} 
  \left\{\begin{array}{l}
           (R,E)\in \bbR_+^2: \\
           R \leq  \min_{s\in\mathcal{S}}\mathbb{I}(P_{X,s},W_{Y|X,s})\\
           E \leq \min_{s\in\mathcal{S}}\phi(P_{X,s})
         \end{array}
  \right\}
\end{align}
{To recover the result of Theorem~\ref{thm:2}, note that for any $P_X\in\mathcal{P}_{\mathcal{X}}$, there exists some type $\widehat{P}_X\in\calP_{\calX}^{(1-\Delta_1-\Delta_2)n}$ such that $\abs{P_X(x)-\widehat{P}_X(x)}\leq\eta$ for any $x\in\calX$ and $\eta>0$ for some $n$ large enough. A similar continuity argument can be used to guarantee that the rate and exponent derived by $\widehat{P}_X$ would be $\xi(\eta)$-close to the result obtained by $P_X$, where the difference $\xi(\eta)$ vanishes with $n$.} By taking the union of all $\{P_{X,s}\}_{s\in\mathcal{S}}$, we conclude that the closure of all achievable regions is at least 
\begin{align}
    \bigcup_{\{P_{X,s''}\}_{s''\in\mathcal{S}}\in \mathcal{P}_{\mathcal{X}}^{ |\mathcal{S}|}} 
  \left\{\begin{array}{l}
           (R,E)\in \bbR_+^2: \\
           R \leq  \min_{s\in\mathcal{S}}\mathbb{I}(P_{X,s},W_{Y|X,s})\\
           E \leq \min_{s\in\mathcal{S}}\phi(P_{X,s})
         \end{array}
  \right\},
\end{align}
which completes the proof.

\subsection{Proof of Theorem~\ref{thm:bi-static-result2}}
We first construct the code by the following steps. 
Fix a specific type $P_X$. The length $n$ codeword corresponds to the message $w\in[1;M]$ is $\mathbf{x}_w$ and is uniformly drawn from the type class $\mathcal{T}^n_{P_X}$. The message decoder and the state estimator are jointly defined as 
\begin{align}
    (\hat{w},\hat{s}) = \argmax_{w\in[1;M],s\in\mathcal{S}} \mathbb{P}\left(\mathbf{y}|S=s,X^n=\mathbf{x}_w\right). 
    \label{eqn:decoder_estimator}
\end{align}
The message decoder $h^{(b)}$ and the state estimator $g^{(b)}$ are then well-defined by (\ref{eqn:decoder_estimator}). Note that the codewords $\{\mathbf{x}_w\}_{w\in[1;M]}$ are random and so is the code $\calC^{(b)}$. Since we have fixed the definition of $h^{(b)}$ and $g^{(b)}$, with a slight abuse of notation, we denote $\calC^{(b)}=\{\bfx_{\ell}\}_{\ell\in[1;M]}$ as the set of codewords in the derivation below. 
We next derive the detection-error exponent when averaging over $\calC^{(b)}$, which codewords are drawn uniformly from $\calT^n_{P_X}$. 
\paragraph{Detection-error Analysis}
Without loss of generality we assume that the message $w=1$ is transmitted and the state $S=s$ for some $s\in\calS$.
The error event is the set of all received $\bfy$ that would result in detection errors and is defined as 
\begin{align*}
    \mathcal{E} \triangleq \{\bfy \in\calY^n:\max_{i\in[1;M]}\bbP(\bfy|S=s,X^n=\bfx_i)<\max_{j\in[1;M]} \bbP(\bfy|S=s',X^n=\bfx_j) \text{ for some }s'\neq s\}. 
\end{align*}
For any $\widehat{P}_{Y|X}\in \mathcal{P}_{\calY|\calX}^n$, it is known from \cite{Thomas_Cover} that the probability of receiving $\bfy\in\mathcal{T}_{\cdot|\bfx_1}(\widehat{P}_{Y|X})$ is upper bounded by 
\begin{align}
    \bbP\left(\bfy \in \calT_{\cdot|\bfx_1 }(\widehat{P}_{Y|X}) \middle| X^n=\bfx_1,S=s\right) \leq e^{-n\mathbb{D}(\widehat{P}_{Y|X}\Vert W_{Y|X,s}|P_X)}. 
    \label{eqn:prof_thm7_0}
\end{align}
For all $\bfy\in \calY^n$ and for any $k\neq 1$, we define the random variable 
\begin{align*}
&I_k(\bfy) \nonumber\\
&\triangleq \mathbf{1}\left(\max_{\ell\neq k}\bbP(\bfy|S=s,X^n=\bfx_{\ell}) \leq  \bbP(\bfy|S=s',X^n=\bfx_{k}) \text{ for some }s'\neq s \right)\\
&= \mathbf{1}\left(\exists s'\neq s \text{ s.t } e^{n\sum_{x,y}P_{X}(x)\widehat{p}_{\bfy|\bfx_k}(y|x)\log W_{Y|X,s'}(y|x)} \geq \max_{\ell\neq k} 
e^{n\sum_{x,y}P_{X}(x)\widehat{p}_{\bfy|\bfx_{\ell}}(y|x)\log W_{Y|X,s}(y|x)}
\right),
\end{align*}
as well as the variable 
\begin{align}
    J_1(\bfy) \triangleq \mathbf{1}\left(\exists s'\neq s \text{ s.t } e^{n\sum_{x,y}P_{X}(x)\widehat{p}_{\bfy|\bfx_1}(y|x)\log W_{Y|X,s'}(y|x)} \geq 
e^{n\sum_{x,y}P_{X}(x)\widehat{p}_{\bfy|\bfx_{1}}(y|x)\log W_{Y|X,s}(y|x)}\right). 
\end{align}
The randomness of the variables $I_k(\bfy)$ and $J_1(\bfy)$ comes from the random coding,  i.e., the type $\hat{p}_{\bfy|\bfx_k}$ is random. $I_k(\bfy)$ is equal to one when there exists some state $s'\neq s$ such that the likelihood calculated according to the state $s'$ and the codeword $k$ is greater than the maximal possible likelihood according to the state $s$.  
Then, the average detection error probability can be upper bounded by the following
\begin{align}
    &\mathbb{E}_{\calC^{(b)}}\left[\bbP\left(\hat{s}\neq s|X^n=\bfx_1,S=s\right)\right] \nonumber\\
    &\leq \frac{1}{|\calT_{P_X}^n|} \sum_{\bfx_1\in \calT_{P_X}^n}\sum_{\widehat{P}_{Y|X}^{(1)}}\sum_{\bfy\in \mathcal{T}_{\cdot|\bfx_1}(\widehat{P}_{Y|X}^{(1)})} \mathbb{P}(\bfy|X^n=\bfx_1,S=s)\mathbb{P}(\text{error}|X^n=\bfx_1,S=s,Y^n=\bfy)\\
    &\leq \frac{1}{|\calT_{P_X}^n|} \sum_{\bfx_1\in \calT_{P_X}^n}\sum_{\widehat{P}_{Y|X}^{(1)}}\sum_{\bfy\in \mathcal{T}_{\cdot|\bfx_1}(\widehat{P}_{Y|X}^{(1)})} \mathbb{P}(\bfy|X^n=\bfx_1,S=s)\sum_{k\neq 1} \mathbb{E}_{\calC^{(b)}\setminus \bfx_1}\left[I_k(\bfy)|X^n=\bfx_1,S=s\right]\nonumber\\
    &+ \frac{1}{|\calT_{P_X}^n|} \sum_{\bfx_1\in \calT_{P_X}^n}\sum_{\widehat{P}_{Y|X}^{(1)}}\sum_{\bfy\in \mathcal{T}_{\cdot|\bfx_1}(\widehat{P}_{Y|X}^{(1)})} \mathbb{P}(\bfy|X^n=\bfx_1,S=s)J_1(\bfy),
    \label{eqn:prof_thm7_1}
\end{align} 
where the expectation is taken over all possibilities of codewords $\{\bfx_{\ell}\}_{\ell\neq 1}$. 
Note that $I_k(\bfy)$ is random because the codewords $\{\bfx_{\ell}\}_{\ell\neq 1}$ are drawn from $\calT^n_{P_X}$ uniformly, and hence, $\{\hat{p}_{\bfy|\bfx_{\ell}}\}_{\ell\neq 1}$ are random, as well.

{
The difficulty of analyzing $I_k(\bfy)$ comes from the term $\max_{\ell\neq k} 
e^{n\sum_{x,y}P_{X}(x)\widehat{p}_{\bfy|\bfx_{\ell}}(y|x)\log W_{Y|X,s}(y|x)}$, which involves a maximization over exponentially many indices $\ell$. Observe that \begin{align}
   \max_{\ell\neq k}e^{n\sum_{x,y}P_{X}(x)\widehat{p}_{\bfy|\bfx_{\ell}}(y|x)\log W_{Y|X,s}(y|x)} = \max_{\ell\neq k} e^{-n\left(\mathbb{D}(\widehat{p}_{\bfy|\bfx_{\ell}}\Vert W_{Y|X,s}|P_X) + \mathbb{H}(\widehat{p}_{\bfy|\bfx_{\ell}}|P_X)\right)}.
   \label{eqn:prof_7_7}
\end{align}
Note that the type $\widehat{p}_{\bfy|\bfx_{1}}$ is fixed when $\bfy\in \mathcal{T}_{\cdot|\bfx_1}(\widehat{P}_{Y|X}^{(1)})$ for some $\widehat{P}_{Y|X}^{(1)} \in \mathcal{P}_{\calY|\calX}^n$. Therefore, in the following we define a set $\mathcal{A}_k$ of conditional types $\left\{\{\widehat{P}_{Y|X}^{(\ell)}\}_{\ell\notin \{1,k\}}\right\}\in \left( \mathcal{P}_{\calY|\calX}^n\right)^{2^{nR}-2}$. 
 such that when $\bfx_{\ell} \in \mathcal{T}_{\bfy|\cdot}(\widehat{P}_{Y|X}^{(\ell)})$ for all $\ell\neq \{1,k\}$, we can control the right-hand side of \eqref{eqn:prof_7_7}. Moreover, $\mathcal{A}_k$ needs to have} the property that $\left\{\{\widehat{p}_{\bfy|\bfx_{\ell}}\}_{\ell\notin \{1,k\}}\right\}\in \mathcal{A}_k$ with high probability. 
Specifically, we define 
\begin{align*}
    &\calA_k(\widehat{P}_{Y|X}^{(1)},P_X,R)\triangleq \left\{\{\widehat{P}_{Y|X}^{(\ell)}\}_{\ell\notin \{1,k\}}\in\left(\calP_{\calY|\calX}^n\right)^{2^{nR}-2}:\forall P''\in \calP_{\calY|\calX}^n \text{ s.t. } \mathbb{I}(P_X,P'') < R\right.\\
    &\quad\left. \text{ and }P_X\circ P''= P_X \circ \widehat{P}_{Y|X}^{(1)} \: \exists j\notin \{1,k\} \text{ s.t } \widehat{P}_{Y|X}^{(j)}=P'', \widehat{P}_{Y|X}^{(j)}\circ P_X = \widehat{P}_{Y|X}^{(1)} \circ P_X \:\forall j\notin\{1,k\}\right\}, 
\end{align*}
as the set of conditional types such that for all $P''\in\calP_{\calY|\calX}^n$ satisfying $\avgI{P_X,P''}<R$ there exists some index $j\notin \{1,k\}$ such that the corresponding conditional type $\widehat{P}_{Y|X}^{(j)}=P''$. The  constraint $\widehat{P}_{Y|X}^{(j)}\circ P_X = \widehat{P}_{Y|X}^{(1)} \circ P_X$ for all $j\notin\{1,k\}$ comes from the fact that when $\bfy\in \mathcal{T}_{\cdot|\bfx_1}(\widehat{P}_{Y|X}^{(1)})$, the type of $\bfy$ is $\widehat{P}_{Y|X}^{(1)} \circ P_X$, and hence $\bbP(\bfx_{\ell}\in \calT_{\bfy|\cdot}(\widehat{P}_{Y|X}^{(\ell)}))=0$ if $\widehat{P}_{Y|X}^{(\ell)}\circ P_X \neq \widehat{P}_{Y|X}^{(1)} \circ P_X$.

Then, by using the law of total probability, we have 
\begin{align*}
    &\mathbb{E}_{\calC^{(b)}\setminus \bfx_1}\left[I_k(\bfy)|X^n=\bfx_1,S=s\right] \nonumber\\
    &= \sum_{\{\widehat{P}_{Y|X}^{(\ell)}\}_{\ell\notin \{1,k\}}} \prod_{\ell\notin \{1,k\}} \bbP(\bfx_{\ell}\in \calT_{\bfy|\cdot}(\widehat{P}_{Y|X}^{(\ell)})) \bbE_{\bfx_k} \left[ I_k\left(\bfy\right)\middle|X^n=\bfx_1,S=s,\widehat{p}_{\bfy|\bfx_{\ell}}=\widehat{P}_{Y|X}^{(\ell)} \forall \ell \notin\{1,k\}\right]\\
    &\leq \sum_{\substack{\{\widehat{P}_{Y|X}^{(\ell)}\}_{\ell\notin \{1,k\}}\\
    \in \calA_k(\widehat{P}_{Y|X}^{(1)},P_X,R)}}  \prod_{\ell\notin \{1,k\}} \bbP(\bfx_{\ell}\in \calT_{\bfy|\cdot}(\widehat{P}_{Y|X}^{(\ell)})) \bbE_{\bfx_k} \left[ I_k\left(\bfy\right)\middle|X^n=\bfx_1,S=s,\widehat{p}_{\bfy|\bfx_{\ell}}=\widehat{P}_{Y|X}^{(\ell)} \forall \ell \notin\{1,k\}\right]\\
    &\quad+ \sum_{\{\widehat{P}_{Y|X}^{(\ell)}\}_{\ell\notin \{1,k\}}\notin \calA_k(\widehat{P}_{Y|X}^{(1)},P_X,R)} \prod_{\ell\notin \{1,k\}} \bbP(\bfx_{\ell}\in \calT_{\bfy|\cdot}(\widehat{P}_{Y|X}^{(\ell)})). 
\end{align*}
Note that when $\{\widehat{P}_{Y|X}^{(\ell)}\}_{\ell\notin \{1,k\}}\in \calA_k(\widehat{P}_{Y|X}^{(1)},P_X,R)$, it holds that 
\begin{align}
    &\min_{\ell\notin \{k,1\} }\bbD(\widehat{P}_{Y|X}^{(\ell)}\Vert W_{Y|X,s}|P_X) + \bbH(\widehat{P}_{Y|X}^{(\ell)}|P_X) \nonumber\\
    &\label{eq:Ak-min}\leq \min_{P''\in \calP_{\calY|\calX}^n:\mathbb{I}(P_X,P'')<R,P_X\circ P''=P_X\circ \widehat{P}_{Y|X}^{(1)}} \bbD(P''\Vert W_{Y|X,s}|P_X) + \bbH(P''|P_X)\\
    &\triangleq \beta_n(\widehat{P}_{Y|X}^{(1)},P_X,R,s),
    \label{eqn:prof_7_8}
\end{align}
{where \eqref{eq:Ak-min} follows since for every $P^{\prime\prime}\in \mathcal{P}_{\calY|\calX}^n$, there is some $\widehat{P}^{(\ell)}_{Y|X}$ such that $\widehat{P}^{(\ell)}_{Y|X}=P''$ by definition of $\mathcal{A}_k.$}
Then, for $\bfy\in \calT_{\cdot|\bfx_1}(\widehat{P}^{(1)}_{Y|X})$, it holds that 
\begin{align*}
&\max_{\ell\neq k} 
e^{n\sum_{x,y}P_{X}(x)\widehat{p}_{\bfy|\bfx_{\ell}}(y|x)\log W_{Y|X,s}(y|x)} \nonumber\\
&\geq \max\left(e^{-n\beta_n(\widehat{P}_{Y|X}^{(1)},P_X,R,s)},e^{-n \bbD( \widehat{P}^{(1)}_{Y|X}\Vert W_{Y|X,s}|P_X) + \bbH( \widehat{P}^{(1)}_{Y|X}|P_X)}\right)
\end{align*}
by applying the right-hand side of \eqref{eqn:prof_7_7} and \eqref{eqn:prof_7_8}.
Therefore, when $\{\widehat{P}_{Y|X}^{(\ell)}\}_{\ell\notin \{1,k\}}\in \calA_k(\widehat{P}_{Y|X}^{(1)},P_X,R)$ and $\bfy\in \calT_{\cdot|\bfx_1}(\widehat{P}_{Y|X}^{(1)})$, the indicator function $I_k$ can be upper bounded by 

\begin{align}
    I_k(\bfy) 
    &\leq \sum_{s'\neq s}\mathbf{1}\left(\bfx_k\in \calT_{\bfy|.}(P_{Y|X}')\text{ for some }P_{Y|X}'\in \calP^n_{s,s'}(\widehat{P}_{Y|X}^{(1)},P_X,R) \right),
\end{align}
where
\begin{align*}
    &\calP_{s,s'}^n(\widehat{P}_{Y|X}^{(1)},P_X,R) \triangleq \Bigg\{P_{Y|X}'\in\calP_{\calY|\calX}^n: e^{n\sum_{x,y}P_{X}(x)P_{Y|X}'(y|x)\log W_{Y|X,s'}(y|x)} \nonumber\\
    &\hspace{0.5cm}\geq e^{-n\min\left(\beta_n(\widehat{P}_{Y|X}^{(1)},P_X,R,s),\mathbb{D}(\widehat{P}_{Y|X}^{(1)}\Vert W_{Y|X,s}|P_X)+\mathbb{H}(\widehat{P}_{Y|X}^{(1)}|P_X)\right)}, P_{Y|X}'\circ P_X = \widehat{P}_{Y|X}^{(1)}\circ P_X\Bigg\}.
\end{align*}
Then, 
\begin{align}
    &\sum_{\substack{\{\widehat{P}_{Y|X}^{(\ell)}\}_{\ell\notin \{1,k\}}\\
    \in \calA_k(\widehat{P}_{Y|X}^{(1)},P_X,R)}} \prod_{\ell\notin \{1,k\}} \bbP(\bfx_{\ell}\in \calT_{\bfy|\cdot}(\widehat{P}_{Y|X}^{(\ell)})) \bbE_{\bfx_k} \left[ I_k\left(\bfy\right)\middle|X^n=\bfx_1,S=s,\widehat{p}_{\bfy|\bfx_{\ell}}=\widehat{P}_{Y|X}^{(\ell)} \forall \ell \notin\{1,k\}\right]\nonumber\\
    &\leq \sum_{s'\neq s}\sum_{P_{Y|X}'\in\calP_{s,s'}^n(\widehat{P}_{Y|X}^{(1)},P_X,R)} \mathbb{P}\left(\bfx_k\in \calT_{\bfy|\cdot}(P_{Y|X}')\right)\\
    &\leq \exp\left(-n\min_{s'\neq s}\min_{P_{Y|X}'\in\calP_{s,s'}^n(\widehat{P}_{Y|X}^{(1)},P_X,R)}(\mathbb{I}(P_X,P_{Y|X}')-o_n(1))\right),
    \label{eqn:prof_thm7_2}
\end{align}
where we have used the fact~\cite[(41)]{GallagerComposition} that $$\P{\bfx_k\in\calT_{\bfy|\cdot}(P^\prime_{Y|X})}=\frac{\abs{\calT^n_{P_XP^\prime_{Y|X}}}}{\abs{\calT^n_{P_X}}\abs{\calT^n_{P_X\circ P^\prime_{Y|X}}}}\leq \exp\left(-n\mathbb{I}(P_X,P^\prime_{Y|X})\right),$$ and $\card{\calP_{s,s'}^n(\widehat{P}^{(1)}_{Y|X},P_X,R)}\leq \text{poly}(n)$.
Moreover, for all $P''\in\calP_{\calY|\calX}^n$ such that $\avgI{P_X,P''}<R$ and for all $\ell\notin\{1,k\}$, it holds that
\begin{align}
\mathbb{P}\left(\bfx_{\ell}\notin \calT_{\bfy|\cdot}(P'')\right) 
&\leq 1-e^{-n(\avgI{P_X,P''}+o_n(1))}
\label{eqn:prof_thm7_2_0}
\\
&\leq e^{-e^{-n(\xi+o_n(1))}},
\label{eqn:prof_thm7_2_1}
\end{align}
for some $\xi<R$, where in \eqref{eqn:prof_thm7_2_0} we lower bound $\P{\bfx_k\in\calT_{\bfy|\cdot}(P^{\prime\prime}_{Y|X})}$ by $e^{-n(\avgI{P_X,P''}+o_n(1))}$~\cite[(14)]{GallagerComposition} and in \eqref{eqn:prof_thm7_2_1} we use the fact that $1-x\leq e^{-x}$ for all $x\in\mathbb{R}$.  Then, 
\begin{align}
\label{eq:type-not-in-Ak}
    \quad\sum_{\substack{\{\widehat{P}_{Y|X}^{(\ell)}\}_{\ell\notin \{1,k\}}\\\notin \calA_k(\widehat{P}_{Y|X}^{(1)},P_X,R)}} \prod_{\ell\notin \{1,k\}} \bbP(\bfx_{\ell}\in \calT_{\bfy|\cdot}(\widehat{P}_{Y|X}^{(\ell)}))&\leq\sum_{\substack{P^{''}\in \calP_{\calY|\calX}^n:\mathbb{I}(P_X,P'')<R,\\P_X\circ P''=P_X\circ \widehat{P}_{Y|X}^{(1)}}}\prod_{\ell\notin \{1,k\}} \bbP(\bfx_{\ell}\notin \calT_{\bfy|\cdot}(P^{''}))\\ 
    &\leq\card{\calP_{\calY|\calX}^n}\left(e^{-e^{-n(\xi-o_n(1))}}\right)^{e^{nR}}\\
    &\leq e^{n o_n(1)}\times e^{-e^{-n(\xi-R-o_n(1))}},
\end{align}
which has the double exponential form and decays much faster than the right-hand side of \eqref{eqn:prof_thm7_2}, {where \eqref{eq:type-not-in-Ak} follows since $\{\widehat{P}_{Y|X}^{(\ell)}\}_{\ell\notin \{1,k\}}\notin \calA_k(\widehat{P}_{Y|X}^{(1)},P_X,R)$ implies that there exists some $P^{''}\in \calP_{\calY|\calX}^n$ such that $\mathbb{I}(P_X,P'')<R,P_X\circ P''=P_X\circ \widehat{P}_{Y|X}^{(1)}$ but no $\bfx_\ell$ lies in $\calT_{\bfy|\cdot}(P^{\prime\prime})$.} Therefore, for $\bfy\in\calT_{\cdot|\bfx_1}(\widehat{P}_{Y|X}^{(1)})$, it holds that 
\begin{align}
    \mathbb{E}_{\calC^{(b)}\setminus \bfx_1}\left[I_k(\bfy)|X^n=\bfx_1,S=s\right] \leq  \exp\left(-n\min_{s'\neq s}\min_{P_{Y|X}'\in\calP_{s,s'}^n(\widehat{P}_{Y|X}^{(1)},P_X,R)}(\mathbb{I}(P_X,P_{Y|X}')-o_n(1))\right).
    \label{eqn:prof_thm7_3}
\end{align}
By applying inequality (\ref{eqn:prof_thm7_3}) to the first term on the right-hand side of (\ref{eqn:prof_thm7_1}) and using (\ref{eqn:prof_thm7_0}), we obtain
\begin{align}
    &\frac{1}{|\calT_{P_X}^n|} \sum_{\bfx_1\in \calT_{P_X}^n}\sum_{\widehat{P}_{Y|X}^{(1)}}\sum_{\bfy\in \mathcal{T}_{\cdot|\bfx_1}(\widehat{P}_{Y|X}^{(1)})} \mathbb{P}(\bfy|X^n=\bfx_1,S=s)\sum_{k\neq 1} \mathbb{E}_{\calC^{(b)}\setminus \bfx_1}\left[I_k(\bfy)|X^n=\bfx_1,S=s\right] \nonumber\\
    &\leq \text{poly}(n) \nonumber\\
    &\times \exp\left(-n\min_{\widehat{P}_{Y|X}^{(1)}}\left( \mathbb{D}(\widehat{P}_{Y|X}^{(1)}\Vert W_{Y|X,s}|P_X)+\min_{s'\neq s}\min_{P_{Y|X}'\in\calP_{s,s'}^n(\widehat{P}_{Y|X}^{(1)},P_X,R)}(\mathbb{I}(P_X,P_{Y|X}')-o_n(1))-R\right)\right).
    \label{eqn:prof_thm7_4}
\end{align}
Besides, the term 
$$
    \frac{1}{|\calT_{P_X}^n|} \sum_{\bfx_1\in \calT_{P_X}^n}\sum_{\widehat{P}_{Y|X}^{(1)}}\sum_{\bfy\in \mathcal{T}_{\cdot|\bfx_1}(\widehat{P}_{Y|X}^{(1)})} \mathbb{P}(\bfy|X^n=\bfx_1,S=s)J_1(\bfy)
$$
is the detection-error probability by using the ML detector when the codeword is known. From Lemma~\ref{lem:1}, we have
\begin{align}
    \frac{1}{|\calT_{P_X}^n|} \sum_{\bfx_1\in \calT_{P_X}^n}\sum_{\widehat{P}_{Y|X}^{(1)}}\sum_{\bfy\in \mathcal{T}_{\cdot|\bfx_1}(\widehat{P}_{Y|X}^{(1)})} \mathbb{P}(\bfy|X^n=\bfx_1,S=s)J_1(\bfy) \leq e^{-n\phi(P_X)}
    \label{eqn:prof_thm7_5}
\end{align}
Combining (\ref{eqn:prof_thm7_1}), (\ref{eqn:prof_thm7_4}), (\ref{eqn:prof_thm7_5}), we have 
\begin{align}
    \mathbb{E}_{\calC^{(b)}}\left[\bbP\left(\hat{s}\neq s|X^n=\bfx_1,S=s\right)\right] \leq e^{-n\min(\phi(P_X),\rho_{\text{joint}}(P_X,R)-\epsilon(n))},
    \label{eqn:prof_thm7_6}
\end{align}
where $\epsilon(n)\rightarrow 0$ when $n\rightarrow\infty$. {The definition of $\rho_{\textnormal{joint}}(P_X,R)$ involves the optimization over conditional distributions, i.e., $\min_{\widehat{P}\in\mathcal{P}_{Y|X}}$ and $\min_{P'\in\mathcal{P}_{s,s'}(\widehat{P},P_X,R)}$, but the result we obtained in \eqref{eqn:prof_thm7_4}
involves optimization over conditional types of length $n$ sequences. However, when $n$ is sufficiently large, any conditional distribution in $\mathcal{P}_{\calY|\calX}$ can be approximated by conditional types in $\mathcal{P}_{\calY|\calX}^n$ with some deviation $\delta(n)$ with $\lim_{n\rightarrow\infty} \delta(n)=0$.  Since all the functions in the definition of $\rho_{\textnormal{joint}}(P_X,R)$ is continuous in the conditional distribution, the difference between the right-hand side of \eqref{eqn:prof_thm7_4} and $\rho_\textnormal{joint}(R_X,R)$ can be bounded by some $\epsilon(\delta(n))$ and again vanishes with $n$.}
The right-hand side of (\ref{eqn:prof_thm7_6}) is irrelevant to the message and the state, and hence,
\begin{align}
    \mathbb{E}_{\calC^{(b)}}\left[\frac{1}{M}\sum_{w}\bbP\left(\hat{s}\neq s|X^n=\bfx_w,S=s\right)\right] \leq e^{-n\min(\phi(P_X),\rho_{\text{joint}}(P_X,R)-\epsilon(\delta))}
\end{align}
for all $s\in\calS$ by the linearity of expectation.
\paragraph{Communication-error Analysis}
By \cite[Lemma 10.1 and Theorem 10.2]{Csiszar2011}, we have that,
\begin{align*}
    \E[\calC^{(b)}]{\max_{s}\max_w \bbP(\hat{w}\neq w|X^n=\bfx_w,S=s)}\leq e^{-n\left(\rho_{\text{succ}}(P_X,R)-\kappa\right)},
\end{align*}
where $\kappa$ vanishes with $n$.
\paragraph{Derandomization and Expurgation}
By the Markov's inequality, we have for any $\zeta_1>0$
\begin{align}
    &\bbP_{\calC^{(b)}}\Bigg(\exists s\in\calS\text{ s.t }\frac{1}{M}\sum_{w}\bbP\left(\hat{s}\neq s|X^n=\bfx_w,S=s\right) > e^{-n(\min(\phi(P_X),\rho_\text{joint}(P_X,R)-\epsilon(\delta))-\zeta_1)}\nonumber\\
    &\hspace{2cm} \text{ or } \max_{s}\max_w \bbP(\hat{w}\neq w|X^n=\bfx_w,S=s) \geq e^{-n(\rho_\text{succ}(P_X,R)-\zeta_2)}
     \Bigg) \nonumber\\
    &\leq |\calS|\frac{\mathbb{E}_{\calC^{(b)}}\left[\frac{1}{M}\sum_{w}\bbP\left(\hat{s}\neq s|X^n=\bfx_w,S=s\right)\right]}{e^{-n\left(\min(\phi(P_X),\rho_\text{joint}(P_X,R)-\epsilon(\delta))-\zeta_1\right)}} + \frac{\bbE_{\calC^{(b)}}\left[\max_{s}\max_w \bbP(\hat{w}\neq w|X^n=\bfx_w,S=s)\right]}{ e^{-n(\rho_\text{succ}(P_X, R)-\zeta_2)}}\\
    &\leq |\calS|e^{-n\zeta_1} + e^{-n\zeta_2},
\end{align}
which goes to zero when $n\rightarrow\infty$. Therefore, there exists some code $\calC^{(b)}$ such that 
\begin{align}
     \max_{s}\frac{1}{M}\sum_{w=1}^M\bbP\left(\hat{s}\neq s|X^n=\bfx_w,S=s\right) < e^{-n\left( \min(\phi(P_X),\rho_\text{joint}(P_X,R)-\epsilon(\delta))-\zeta_1\right)},\\
     \max_{s}\max_w \bbP(\hat{w}\neq w|X^n=\bfx_w,S=s) < e^{-n(\rho_\text{succ}(P_X, R)-\zeta_2)}
\end{align}
when $n$ is sufficiently large. By the codebook expurgation argument, there exists a code $\Bar{C}^{(b)}\triangleq\{\Bar{\bfx}_w\}\subset{\calC^{(b)}}$ with size $|\Bar{C}^{(b)}|=|\calC^{(b)}|/2$ such that 
\begin{align}
     \max_{s}\max_{w}\bbP\left(\hat{s}\neq s|X^n=\Bar{\bfx}_w,S=s\right) < 2e^{-n\left( \min(\phi(P_X),\rho_\text{joint}(P_X,R)-\epsilon(\delta))-\zeta_1\right)},\\
      \max_{s}\max_w \bbP(\hat{w}\neq w|X^n=\bfx_w,S=s) < e^{-n(\rho_\text{succ}(P_X, R)-\zeta_2)}.
\end{align}
For every $R<\min_s\mathbb{I}(P_X,W_{Y|X,s})$, there exists some $\zeta_2$ small enough such that $\rho_\text{succ}(P_X,R)-\zeta_2>0$. Therefore, for every $R<\min_s\mathbb{I}(P_X,W_{Y|X,s})$ there exists a code $\Bar{C}^{(b)}$ with rate $R-o_n(1)$ such that 
\begin{align}
    -\frac{1}{n} \log \max_{s}\max_{w}\bbP\left(\hat{s}\neq s|X^n=\Bar{\bfx}_w,S=s\right) \geq  \min(\phi(P_X),\rho_\text{joint}(P_X,R))-\epsilon,\\
    \max_{s}\max_w \bbP(\hat{w}\neq w|X^n=\Bar{\bfx}_w,S=s) < \epsilon
\end{align}
for any $\epsilon>0$ whenever $n$ is sufficiently large. 
The theorem is proved by taking the union over all possible $P_X$. 

\subsection{Proof of Corollary~\ref{Cor:3}}
Recall that the definition of $\rho_{\text{joint}}(P_X,R)$ is
\begin{align}
    \rho_{\text{joint}}(P_X, R) \triangleq \min_{s\in\calS}\min_{\widehat{P}}\left( \mathbb{D}(\widehat{P}\Vert W_{Y|X,s}|P_X)+\left|\min_{s'\neq s}\min_{P'\in\calP_{s,s'}(\widehat{P},P_X,R)}\mathbb{I}(P_X,P')-R\right|^{+}\right).
\end{align}
For any $\widehat{P}$, $P_X$, $s\in\calS$ and $s'\neq s$, the definition of $\calP_{s,s'}(\widehat{P},P_X,R)$ is
\begin{align*}
    &\calP_{s,s'}(\widehat{P},P_X,R) \triangleq \Bigg\{P': 
    \mathbb{D}(P'\Vert W_{Y|X,s'}|P_X) + \mathbb{H}(P'|P_X) \nonumber\\
    &\leq \min\left(\beta(\widehat{P},P_X,R,s),\mathbb{D}(\widehat{P}\Vert W_{Y|X,s}|P_X)+\mathbb{H}(\widehat{P}|P_X)\right), P_X\circ P' =  P_X\circ \widehat{P}\Bigg\}. 
\end{align*}
The value of $\min\left(\beta(\widehat{P},P_X,R,s),\mathbb{D}(\widehat{P}\Vert W_{Y|X,s}|P_X)+\mathbb{H}(\widehat{P}|P_X)\right)$ depends on $\widehat{P}$, $s$ and $P_X$. Therefore, for each $s\in\calS$ and $P_X$, we partition the set $\calP_{\calY|\calX}$ into $\calP_{\calY|\calX}'(P_X,s)$ and $\calP_{\calY|\calX}''(P_X,s)$, where 
\begin{align*}
    \calP_{\calY|\calX}'(P_X,R,s) &\triangleq \{\widehat{P}\in\calP_{\calY|\calX}:\max_{s''\in\calS}\beta(\widehat{P},P_X,R,s'') \geq \mathbb{D}(\widehat{P}\Vert W_{Y|X,s}|P_X)+\mathbb{H}(\widehat{P}|P_X)\}\\
    \calP_{\calY|\calX}''(P_X,R,s) &\triangleq \{\widehat{P}\in\calP_{\calY|\calX}:\max_{s''\in\calS}\beta(\widehat{P},P_X,R,s'') < \mathbb{D}(\widehat{P}\Vert W_{Y|X,s}|P_X)+\mathbb{H}(\widehat{P}|P_X)\}. 
\end{align*}
For each $s\in\calS$, and $P_X$, we define 
\begin{align}
    \widehat{\rho}(P_X, R,s,\widehat{P}) \triangleq 
    \mathbb{D}(\widehat{P}\Vert W_{Y|X,s}|P_X)+\left|\min_{s'\neq s}\min_{P'\in\calP_{s,s'}(\widehat{P},P_X,R)}\mathbb{I}(P_X,P')-R\right|^{+}. 
\end{align}
The exponent $\rho_{\text{joint}}(P_X, R)$ is then 
\begin{align}
    \rho_{\text{joint}}(P_X, R) = \min_{s\in\calS}\min\left(\min_{\widehat{P}\in\calP_{\calY|\calX}'(P_X,R,s)}\widehat{\rho}(P_X, R,s,\widehat{P}),\min_{\widehat{P}\in\calP_{\calY|\calX}''(P_X,R,s)}\widehat{\rho}(P_X,R,s,\widehat{P})\right). 
\end{align}
We first analyze the term $\min_s\min_{\widehat{P}\in\calP_{\calY|\calX}'(P_X,R,s)}\widehat{\rho}(P_X, R,s,\widehat{P})$. 
Let the arguments of the minimization of the term 
$$
    \min_{\widehat{P}\in\calP_{\calY|\calX}'(P_X,R,s)} \mathbb{D}(\widehat{P}\Vert W_{Y|X,s}|P_X)+\left|\min_{s'\neq s}\min_{P'\in \calP_{s,s'}(\widehat{P},P_X,R)}\mathbb{I}(P_X,P')-R\right|^{+}. 
$$
be achieved by the tuple $(\widehat{P},s',P')$, where $\widehat{P}$, $s'$ and $P'$ are the minimizer corresponding to $\min_{\widehat{P}\in\calP_{\calY|\calX}'(P_X,R,s)}$, $\min_{s'\neq s}$ and $\min_{P'\in \calP_{s,s'}(\widehat{P},P_X,R)}$, respectively. Given the tuple $(\widehat{P},s',P')$, if $\mathbb{I}(P_X,P')\geq \mathbb{I}(P_X,\widehat{P})$, then 
\begin{align*}
    & \min_{\widehat{P}\in\calP_{\calY|\calX}'(P_X,R,s)} \mathbb{D}(\widehat{P}\Vert W_{Y|X,s}|P_X)+\left|\min_{s'\neq s}\min_{P'\in \calP_{s,s'}(\widehat{P},P_X,R)}\mathbb{I}(P_X,P')-R\right|^{+}\\
    &= \mathbb{D}(\widehat{P}\Vert W_{Y|X,s}|P_X)+\left|\mathbb{I}(P_X,P')-R\right|^{+}\\
    &\geq \mathbb{D}(\widehat{P}\Vert W_{Y|X,s}|P_X)+\left|\mathbb{I}(P_X,\widehat{P})-R\right|^{+}\\
    &\geq \min_{\widehat{P}\in\calP_{\calY|\calX}'(P_X,R,s)} \mathbb{D}(\widehat{P}\Vert W_{Y|X,s}|P_X)+\left|\mathbb{I}(P_X,\widehat{P})-R\right|^{+}.
\end{align*}
On the other hand, if the minimizer $P'$ satisfies $\mathbb{I}(P_X,P')<\mathbb{I}(P_X,\widehat{P})$, then by definition of $\calP_{s,s'}(\widehat{P},P_X)$ and the fact that $\widehat{P}\in\calP_{\calY|\calX}'(P_X,R,s)$ it holds that 
$$
    \mathbb{D}(\widehat{P}\Vert W_{Y|X,s}|P_X) \geq \mathbb{D}(P'\Vert W_{Y|X,s'}|P_X)
$$
Moreover, since $P'\in \calP_{s,s'}(\widehat{P},P_X,R)$, it holds that 
\begin{align}
    \mathbb{D}(P'\Vert W_{Y|X,s'}|P_X) + \mathbb{H}(P'|P_X) &\leq \beta(\widehat{P},P_X,R,s) \\
    &= \beta(P',P_X,R,s)\label{eqn:cor_3_1}\\
    &\leq \max_{s''\in\calS} \beta(P',P_X,R,s'')
    \label{eqn:cor_3_2},
\end{align}
where in \eqref{eqn:cor_3_1} we use the fact that $P_X\circ P'=P_X\circ \widehat{P}$. 
Therefore, by definition of $\calP'_{\calY|\calX}$ and \eqref{eqn:cor_3_2}, it holds that $P'\in \calP_{\calY|\calX}'(P_X,R,s')$. 
Then, 
\begin{align}
    & \min_{\widehat{P}\in\calP_{\calY|\calX}'(P_X,R,s)} \mathbb{D}(\widehat{P}\Vert W_{Y|X,s}|P_X)+\left|\min_{s'\neq s}\min_{P'\in \calP_{s,s'}(\widehat{P},P_X)}\mathbb{I}(P_X,P')-R\right|^{+}\\
    &=\mathbb{D}(\widehat{P}\Vert W_{Y|X,s}|P_X)+\left|\mathbb{I}(P_X,P')-R\right|^{+}\\
    &\geq \mathbb{D}(P'\Vert W_{Y|X,s'}|P_X)+\left|\mathbb{I}(P_X,P')-R\right|^{+}\\
    &\geq \min_{P'\in \calP_{\calY|\calX}'(P_X,R,s')}\mathbb{D}(P'\Vert W_{Y|X,s'}|P_X)+\left|\mathbb{I}(P_X,P')-R\right|^{+}.
\end{align}
Taking the minimum over both cases of $\mathbb{I}(P_X,P')\geq \mathbb{I}(P_X,\widehat{P})$ and $\mathbb{I}(P_X,P')<\mathbb{I}(P_X,\widehat{P})$ and all possible $s\in\calS$, we obtain that 
\begin{align}
    \min_s\min_{\widehat{P}\in\calP_{\calY|\calX}'(P_X,R,s)}\widehat{\rho}(P_X,R,s,\widehat{P}) \geq  \min_s\min_{\widehat{P}\in\calP_{\calY|\calX}'(P_X,R,s)} \mathbb{D}(\widehat{P}\Vert W_{Y|X,s}|P_X)+\left|\mathbb{I}(P_X,\widehat{P})-R\right|^{+}.
    \label{eqn:cor_1}
\end{align}
Our next step is to analyze $\min_s\min_{\widehat{P}\in\calP_{\calY|\calX}''(P_X,R,s)}\widehat{\rho}(P_X, R,s,\widehat{P})$. We assume again that the minimization of the term 
$$
    \min_{\widehat{P}\in\calP_{\calY|\calX}''(P_X,R,s)} \mathbb{D}(\widehat{P}\Vert W_{Y|X,s}|P_X)+\left|\min_{s'\neq s}\min_{P'\in \calP_{s,s'}(\widehat{P},P_X,R)}\mathbb{I}(P_X,P')-R\right|^{+} 
$$
is achieved by the tuple $(\widehat{P},s',P')$. \\
Given the tuple $(\widehat{P},s',P')$, if $\mathbb{D}(\widehat{P}\Vert W_{Y|X,s}|P_X) \geq \mathbb{D}(P'\Vert W_{Y|X,s'}|P_X)$, then 
\begin{align}
    &\min_{\widehat{P}\in\calP_{\calY|\calX}''(P_X,R,s)} \mathbb{D}(\widehat{P}\Vert W_{Y|X,s}|P_X)+\left|\min_{s'\neq s}\min_{P'\in \calP_{s,s'}(\widehat{P},P_X,R)}\mathbb{I}(P_X,P')-R\right|^{+} \\
    &=  \mathbb{D}(\widehat{P}\Vert W_{Y|X,s}|P_X)+\left|\mathbb{I}(P_X,P')-R\right|^{+}\\
    &\geq \mathbb{D}(P'\Vert W_{Y|X,s'}|P_X)+\left|\mathbb{I}(P_X,P')-R\right|^{+}\\
    &\geq \min_{P'\in \calP_{\calY|\calX}'(P_X,R,s')}\mathbb{D}(P'\Vert W_{Y|X,s'}|P_X)+\left|\mathbb{I}(P_X,P')-R\right|^{+},
    \label{eqn:cor_2}
\end{align}
where (\ref{eqn:cor_2}) follows from the fact that $P'\in\calP'_{\calY|\calX}(P_X,R,s')$ and has the same form as the right-hand side of (\ref{eqn:cor_1}). On the other hand, if $\mathbb{D}(\widehat{P}\Vert W_{Y|X,s}|P_X) < \mathbb{D}(P'\Vert W_{Y|X,s'}|P_X)$, we have

\begin{align}
    &\min_{\widehat{P}\in\calP_{\calY|\calX}''(P_X,R,s)} \mathbb{D}(\widehat{P}\Vert W_{Y|X,s}|P_X)+\left|\min_{s'\neq s}\min_{P'\in \calP_{s,s'}(\widehat{P},P_X)}\mathbb{I}(P_X,P')-R\right|^{+} \nonumber\\
    &=  \mathbb{D}(\widehat{P}\Vert W_{Y|X,s}|P_X)+\left|\mathbb{I}(P_X,P')-R\right|^{+}\nonumber\\
    &\geq \mathbb{D}(\widehat{P}\Vert W_{Y|X,s}|P_X)+\left|\mathbb{D}(P'\Vert W_{Y|X,s'}|P_X) + \mathbb{H}(P_X\circ P') - \mathbb{D}(P'\Vert W_{Y|X,s'}|P_X)-\mathbb{H}(P'|P_X)-R\right|^{+}\nonumber\\
    &\labrel\geq{eq:P''-set-ineq}  \min_{\widehat{P}\in\calP_{\calY|\calX}''(P_X,R,s)} \mathbb{D}(\widehat{P}\Vert W_{Y|X,s}|P_X) +\Bigg|\mathbb{D}(\widehat{P}\Vert W_{Y|X,s}|P_X) + \mathbb{H}(P_X\circ \widehat{P}) - \beta(\widehat{P},P_X,R,s)-R\Bigg|^{+},
    \label{eqn:cor_3}
\end{align}
where \eqref{eq:P''-set-ineq} follows since $P'\in\calP_{s, s'}(\widehat{P},P_X,R)$ implies that $-\mathbb{D}(P'\Vert W_{Y|X,s'}|P_X)-\mathbb{H}(P'|P_X) \geq -\beta(\widehat{P},P_X,R,s)$ and we use the fact that $P_X\circ P'=P_X\circ \widehat{P}$.

Finally, by taking the minimum over all $s$ and applying the fact that $\rho_{\text{joint}}(P_X,R)$ is the minimum of the right-hand sides of (\ref{eqn:cor_1}) and (\ref{eqn:cor_3}), we conclude that 
\begin{align}
    \rho_{\text{joint}}(P_X,R) \geq \min_s\min\left(\min_{\widehat{P}\in\calP_{\calY|\calX}'(P_X,R,s)}\gamma_1(\widehat{P},R,s),\min_{\widehat{P}\in\calP_{\calY|\calX}''(P_X,R,s)}\gamma_2(\widehat{P},R,s)\right),
\end{align}
where for all $\widehat{P}$, $R$  and $s$, 
\begin{align*}
    \gamma_1(\widehat{P},R,s) &\triangleq  \mathbb{D}(\widehat{P}\Vert W_{Y|X,s}|P_X)+\left|\mathbb{I}(P_X,\widehat{P})-R\right|^{+}\\
    \gamma_2(\widehat{P},R,s) &\triangleq \mathbb{D}(\widehat{P}\Vert W_{Y|X,s}|P_X) +\Bigg|\mathbb{D}(\widehat{P}\Vert W_{Y|X,s}|P_X)  + \mathbb{H}(P_X\circ\widehat{P}) - \beta(P_X,R,s) -R\Bigg|^{+}. 
\end{align*}

\bibliographystyle{IEEEtran}
\bibliography{JointCommSensing-final-one-col.bib}

\end{document}